\documentclass[conference]{IEEEtran}
\usepackage{amsmath}
\usepackage{amsthm}
\usepackage{amsfonts}
\usepackage{amssymb}
\usepackage{multicols}
\usepackage{authblk}
\usepackage{graphicx}
\usepackage{epstopdf}
\usepackage{mathtools}

\theoremstyle{definition}
\newtheorem{thm}{Theorem}
\newtheorem{lem}{Lemma}

\newtheorem{cor}{Corollary}
\newtheorem{defn}{Definition}

\newtheorem{exmp}{Example}

\newtheorem{obs}{Observation}

\newtheorem{note}{Note}

\DeclareRobustCommand{\rchi}{{\mathpalette\irchi\relax}}
\newcommand{\irchi}[2]{\raisebox{\depth}{$#1\chi$}}

\title{Optimal Scalar Linear Codes for Some Classes of The Two-Sender Groupcast Index Coding Problem}
\author{Chinmayananda Arunachala  and B. Sundar Rajan. \\
Department of Electrical Communication Engineering, Indian Institute of Science, Bengaluru 560012, KA, India. \\
E-mail: \{chinmayanand,bsrajan\}@iisc.ac.in}
\date{May 2017}
\begin{document}
\maketitle
\begin{abstract}
The two-sender groupcast index coding problem (TGICP) consists of a set of receivers, where all the messages demanded by the set of receivers are distributed among the two senders. The senders can possibly have a set of messages in common. Each message can be demanded by more than one receiver. Each receiver has a subset of messages (known as its side information) and demands a message it does not have. The objective is to design scalar linear codes at the senders with the  minimum aggregate code length such that all the receivers are able to decode their demands, by  leveraging the knowledge of the side information of all the receivers. In this work, optimal scalar linear codes of three sub-problems (considered as single-sender groupcast index coding problems (SGICPs)) of the TGICP are used to construct optimal scalar linear codes for some classes of the TGICP. We introduce the notion of joint extensions of a finite number of SGICPs, which generalizes the notion of extensions of a single SGICP introduced in a prior work. An SGICP $\mathcal{I}_E$ is said to be a joint extension of a finite number of SGICPs if all the SGICPs are disjoint sub-problems of $\mathcal{I}_E$. We identify a class of joint extensions, where optimal scalar linear codes of the joint extensions can be constructed using those of the sub-problems. We then construct scalar linear codes for some classes of the TGICP, when one or more sub-problems of the TGICP belong to the above identified class of joint extensions, and provide some necessary conditions for the optimality of the construction.
   
\end{abstract}
%%%%%%%%%%%%%%%%%%%%%%%%%%%%%%%%%%%%%%%%%%%%%%%%%%%%%%%%

\section{INTRODUCTION}
\par The classical index coding problem (ICP) introduced in \cite{BK} is a source coding problem consisting of a set of receivers, each demanding a single message and having a subset of other messages as its side information. The sender exploits the knowledge of the side information at all the receivers to reduce the number of transmissions required for all the receivers to decode their demands. The multi-sender ICP (first studied in \cite{SUOH}) arises as an extension of the classical ICP in many practical scenarios. The demands of all the receivers are distributed among multiple senders, with possibly some messages in common with senders belonging to any subset of senders. The senders collectively have all the demands of the receivers. All the receivers can receive the transmissions from all the senders. Transmissions from any two senders are orthogonal in time. All the senders are able to communicate among themselves and collaboratively transmit coded messages to reduce the end-to-end latency and network-traffic load as illustrated in \cite{LOJ}, \cite{KAC}. Each sender is naturally restricted to have a subset of demanded messages due to constraints like limited storage (as illustrated in \cite{LOJ}), erroneous reception of messages (as illustrated in \cite{CVBSR1}), and node failures as seen in distributed storage networks where data is stored over different nodes \cite{luo2016coded}, \cite{xiang2016joint}.

A class of the multi-sender ICP was studied, where each receiver knows a unique message which is not known to any other receiver \cite{SUOH}. This class was analyzed using  graph theoretic techniques and optimal codes were provided for a special sub-class. A rank-minimization framework was proposed along with a heuristic algorithm to obtain sub-optimal index codes (in general) for another class of the multi-sender ICP \cite{LOJ}. For this class, each receiver demands a message that is not demanded by any other receiver and has a subset of other messages as side information. The most general class of the multi-sender ICP where any message can be demanded by any number of receivers was studied in \cite{KN2}. The problem of finding the optimal scalar linear code length of the multi-sender ICP was related to a matrix completion problem with rank minimization. Some variants of the multi-sender ICP have been studied with an independent link of fixed capacity existing between every sender and receiver \cite{PA}-\cite{MOJ}. Inner and outer bounds on the capacity region of these variants were established.

The two-sender ICP has been studied in several works as a fundamental multi-sender ICP \cite{COJ}-\cite{CVBSR2}. Some single-sender index coding schemes based on graph theory were extended to the two-sender unicast ICP (TUICP) \cite{COJ}. In the TUICP, each receiver demands a unique message which is not demanded by any other receiver. Optimal broadcast rate  with finite length messages (total number of coded bits transmitted per message bit) of some special classes of the TUICP, and its limiting value as the message length tends to infinity (called the optimal broadcast rate) were studied \cite{CTLO}, \cite{CVBSR2}. Code constructions were provided for some classes of the TUICP using optimal codes of three single-sender sub-problems \cite{CTLO}-\cite{CVBSR2}. For some classes, a new graph coloring technique (termed as two-sender graph coloring) was proposed, and employed to obtain the optimal broadcast rate with finite length messages \cite{CTLO}, \cite{CVBSR2}. Optimal broadcast rate with finite length messages and its limiting value were established for some classes of the TUICP considering only linear encoding schemes at the senders \cite{CVBSR1}. The proof techniques used in \cite{CVBSR1} were shown to be useful in obtaining optimal results for some classes of the multi-sender ICP with any number of senders.    
 
We consider the two-sender groupcast ICP (TGICP) in this work, where any receiver can demand any message and can have any subset of other messages as side information. We first divide the TGICP into three disjoint single-sender sub-problems (no common messages between any two of the sub-problems) based on the demands of the receivers and the availability of messages with the senders. We classify the TGICP based on the relations between the three single-sender sub-problems of the TGICP (Section II). This classification extends the classification given in \cite{CTLO} for the TUICP (which is a special case of the TGICP). Optimal scalar linear codes are constructed using those of the three sub-problems for some clasees of the TGICP (Sections III and V).  We introduce the notion of joint extensions of any finite number of single-sender groupcast ICPs (SGICPs). Any jointly extended problem obtained from a finite number of SGICPs (dealt in this work), consists of all the SGICPs as disjoint sub-problems (no messages in common between any two of the sub-problems). A class of joint extensions is identified and scalar linear codes for this class of jointly extended problems are constructed  using those of its sub-problems (Section IV). A necessary condition for the optimality of the constructed codes is given.  This result is used to construct scalar linear codes for some classes of the TGICP (Section V). Necessary conditions for the optimality of the constructed codes are also given.  

\par The key results of this paper are highlighted below. 

\begin{itemize}
	\item  Optimal scalar linear codes for some classes of the TGICP are constructed using optimal scalar linear codes of three single-sender sub-problems. According to the authors' knowledge, this is the first work providing optimal scalar linear codes for some classes of the TGICP (which are not TUICPs). Thus, the problem of constructing optimal scalar linear codes for these classes of the TGICP, is reduced to the problem of optimal code construction for the SGICP.  Optimal linear broadcast rates with finite length messages for the TUICP presented in \cite{CVBSR1} is a special case of the results presented in this work.
	\item  The notion of joint extensions of any finite number of SGICPs is introduced, which generalizes the notion of extensions of any SGICP, introduced in \cite{PK}. A class of joint extensions is identified for which optimal scalar linear codes for the jointly extended problems can be constructed using those of the sub-problems. Two classes of rank-invariant extensions presented in \cite{PK} are shown to be special cases of the class of joint extensions addressed in this paper. Using the result related to the class of joint extensions solved in this paper, optimal scalar linear codes for some SGICPs can be constructed using known optimal scalar linear codes of their sub-problems.  	
	\item Scalar linear codes are constructed for some classes of the TGICP, where one or more sub-problems belong to the class of joint extensions solved in this paper. Necessary conditions for the optimality of the constructed codes are also given. 
 \end{itemize}

\par The remainder of the paper is organized as follows. Section II consists of problem formulation and  classification of the TGICP. Section III provides construction of optimal scalar linear codes for some classes of the TGICP. Section IV defines joint extensions and provides optimal codes for a special class of joint extensions. Section V provides code  constructions for some classes of the TGICP using the results of Section IV. Section VI concludes the paper.

\par \emph{Notations:} Matrices and vectors are denoted by bold uppercase and bold lowercase letters respectively. For any positive integer $m$, $[m]  \triangleq \{1,...,m\}$. The finite field with $q$ elements is denoted by $\mathbb{F}_q$. The vector space of all $n \times d $ matrices over $\mathbb{F}_q$ is denoted by  $\mathbb{F}_q^{n \times d}$. For any matrix ${\bf{M}}\in \mathbb{F}_q^{n \times d}$, ${\bf{M}}_{\mathcal{A}}$ denotes the matrix obtained by stacking (one upon the other in the increasing order of row indices) the rows of ${\bf{M}}$ indexed by the elements in $\mathcal{A}$ for any $\mathcal{A} \subseteq [n]$. The row space of ${\bf{M}}$ is denoted by $\langle {\bf{M}} \rangle$. For any matrix ${\bf{M}}$ over $\mathbb{F}_q$, the rank of ${\bf{M}}$ over $\mathbb{F}_q$ is denoted as $\text{rk}_{q}({\bf{M}})$. The transpose of ${\bf{M}}$ is denoted by ${\bf{M}}^T$. A matrix obtained by deleting some of the rows and/or some of the columns of ${\bf{M}}$ is said to be a submatrix of ${\bf{M}}$. A set of submatrices of a given matrix are said to be disjoint, if no two of the submatrices have common elements indexed by the same ordered pair in the given matrix. The matrix consisting of all $x$'s is denoted by ${\bf{X}}$. The size of ${\bf{X}}$ is clear from the context whenever it is not explicitly stated. Similarly, the matrix with all entries being $0$ is denoted by ${\bf{0}}$. An $a \times b$ matrix with all entries being $0$ is denoted as ${\bf{0}}_{a \times b}$. An $r \times r$ identity matrix is denoted by ${\bf{I}}_{r \times r}$.

A directed graph (also called digraph) given by $\mathcal{D}=(\mathcal{V}(\mathcal{D}),\mathcal{E}(\mathcal{D}))$, consists of a set of vertices $\mathcal{V}(\mathcal{D})$, and a set of edges $\mathcal{E}(\mathcal{D})$ which is a set of ordered pairs of vertices. A cycle in a digraph $\mathcal{D}$ is a sequence of distinct vertices $(v_{i_1},...,v_{i_c})$ such that $(v_{i_s},v_{i_{s+1}}) \in \mathcal{E}(\mathcal{D})$ for all $s \in [c-1]$ and $(v_{i_c},v_{i_1}) \in \mathcal{E}(\mathcal{D})$. A digraph is called acyclic if it does not contain any cycles. 

\section{Problem Formulation}
\par In this section, we formulate the TGICP and establish the required notation and definitions.

An instance of the TGICP consists of two senders collectively having $m$ independent message symbols given by the set $\{{\bf{x}}_1,{\bf{x}}_2,...,{\bf{x}}_{m}\}$, where ${\bf{x}}_i \in \mathbb{F}_q^{t \times 1}$, $\forall i \in [m]$, and $t \geq 1$ is an integer. The $j$th sender denoted by $\mathcal{S}_{j}$, $j \in \{1,2\}$, has a set of messages with indices given by the set $\mathcal{M}_{j}$, such that $\mathcal{M}_{j} \subseteq [m]$ and $\mathcal{M}_{1} \cup \mathcal{M}_{2}=[m]$. Each sender knows the identity of the messages (indices of the message symbols) available with the other. Each sender transmits over a noiseless broadcast channel which carries symbols from $\mathbb{F}_q$. Transmissions from different senders are orthogonal in time. There are $n$ receivers receiving all the transmissions from both the senders. Without loss of generality, we assume that each receiver demands a single message and that every message is demanded by at least one receiver. Let ${\bf{x}}_{\mathcal{A}} \triangleq ({\bf{x}}_i)_{i \in \mathcal{A}}$ for any non-empty $\mathcal{A} \subseteq [m]$. Note that ${\bf{x}}_{\mathcal{A}}$ is the vector obtained by the  concatenation of all the message symbols with indices given by the set $\mathcal{A}$. The $i$th receiver knows a vector of message symbols (said to be its side information)  given by ${\bf{x}}_{\rchi_i}$, where $\rchi_i \subset [m]$, and demands ${\bf{x}}_{f(i)}$, where the mapping $f : [n] \rightarrow [m]$ satisfies $f(i) \notin \rchi_i$, $i \in [n]$. The set $\{\rchi_i\}_{i \in [n]}$ is known to both the senders. The objective of the TGICP is to design a coding scheme at each sender such that all the receivers are able to decode their demands with the least aggregate number of  transmissions from the two senders.
We formalize this objective in the following.

For a given instance of the TGICP, an index code over $\mathbb{F}_q$ is given by encoding functions at the two senders such that there exists a decoding function at each receiver as described in the following. An encoding function at  the sender $\mathcal{S}_{j}$, $j \in \{1,2\}$, is given by $\mathbb{E}_{j} :\mathbb{F}_{q}^{|\mathcal{M}_{j}|t \times 1} \rightarrow   \mathbb{F}_{q}^{l_{j} \times 1}$,
 where $l_j$ is the length of the codeword $\mathbb{E}_j({\bf{x}}_{\mathcal{M}_j})$ transmitted by $\mathcal{S}_{j}$. A decoding function at the $i$th receiver is  given by a map $\mathbb{D}_{i} :\mathbb{F}_{q}^{(|\rchi_{i}|t+l_{1}+l_{2}) \times 1} \rightarrow   \mathbb{F}_{q}^{t \times 1}$, such that $\forall {\bf{x}}_{[m]} \in \mathbb{F}_q^{mt \times 1}$, $\mathbb{D}_i(\mathbb{E}_j({\bf{x}}_{\mathcal{M}_1}),\mathbb{E}_j({\bf{x}}_{\mathcal{M}_2}),{\bf{x}}_{\rchi_i})={\bf{x}}_{f(i)}$, $i \in [n]$. That is, every $i$th receiver can decode ${\bf{x}}_{f(i)}$ using its side information ${\bf{x}}_{\rchi_i}$ and the received codewords, $i \in [n]$. The code length of the index code given by the encoding functions $\{\mathbb{E}_j\}_{j \in \{1,2\}}$ is $l_1 + l_2$. The objective of the TGICP is to find the minimum value of $l_1+l_2$ and a pair of encoding functions which collectively achieve this code length. An index code for the TGICP is said to be linear if both the encoding functions are linear transformations over $\mathbb{F}_q$. For a linear index code, the codeword $\mathbb{E}_j({\bf{x}}_{\mathcal{M}_j})$ can be expressed as  ${\bf{G}}^{(j)}{\bf{x}}_{\mathcal{M}_j}$, where  ${\bf{G}}^{(j)} \in \mathbb{F}_{q}^{l_j \times t|\mathcal{M}_j|}$ is called an encoding matrix at $\mathcal{S}_j$ for the given code, $j \in \{1,2\}$. All the receivers are assumed to know $\{{\bf{G}}^{(i)}\}_{i \in \{1,2\}}$. If $t=1$, the code is said to be a scalar index code. Otherwise, it is said to be a vector index code. For a given instance $\mathcal{I}$ of the TGICP, optimal scalar linear code length is denoted as $l^*_q(\mathcal{I})$. For the single-sender groupcast ICP (SGICP), without loss of generality we assume that only $\mathcal{S}_1$ exists and $\mathcal{M}_{1}=[m]$. The unicast ICP (single-sender or two-sender) is a special case of the corresponding groupcast ICP with $m=n$. Hence no message in a unicast ICP is demanded by two receivers.

We construct optimal scalar linear codes for some classes of the TGICP using those of its sub-problems,  considering the sub-problems as SGICPs. We now recapitulate some required results related to the SGICP.  
 
An instance of the SGICP is described by a matrix called the fitting matrix of the SGICP consisting of known entries ($1$'s and $0$'s) and unknown entries (denoted as $x$'s)  \cite{PK}. The notion of the fitting matrix was introduced in \cite{BY} for the single-sender unicast ICP and later extended to the SGICP \cite{PK}. Each row of the fitting matrix represents a receiver and each column represents a message, and is described as follows. 

\begin{defn}[Fitting Matrix, \cite{PK}]
	An SGICP with $m$ messages and $n$ receivers is described by an $n \times m$ matrix ${\bf{F}}_x$ (subscript `$x$' indicates presence of unknown entries), called the fitting matrix of the SGICP whose $(i,j)$th entry is given by, 
	\[ 
	[{\bf{F}}_x]_{i,j}=
	\begin{cases}
	x               & \text{if} \ j \in \rchi_i,\\
	1               & \text{if} \ f(i)=j,\\
	0               & \text{otherwise}.
	\end{cases}
	\]
	$\forall$ $i \in [n],j \in [m]$.
\end{defn}

For a given fitting matrix ${\bf{F}}_x$, we say ${\bf{F}}$ completes ${\bf{F}}_x$ and denote it as ${\bf{F}} \approx {\bf{F}}_x$, if ${\bf{F}}$ is obtained from ${\bf{F}}_x$ by replacing all the unknowns with arbitrary elements from the field. For the single-sender unicast ICP described by the fitting matrix ${\bf{F}}_x$, the optimal scalar linear code length over $\mathbb{F}_q$ has been shown to be the minimum rank of a matrix ${\bf{F}}$ such that ${\bf{F}} \approx {\bf{F}}_x$  \cite{BY}. This minimum rank is denoted as $\text{mrk}_{q}({\bf{F}}_x)$. The optimal scalar linear code length of the SGICP described by ${\bf{F}}_x$ has also been shown to be equal to $\text{mrk}_{q}({\bf{F}}_x)$ in Corollary $4.7$, \cite{DSC2}.

The two-sender unicast ICP (TUICP) was classified  based on the relation between three single-sender sub-problems  \cite{CTLO}. We now extend this classification to the TGICP. 

Let $\mathcal{P}_1=\mathcal{M}_1  \setminus \mathcal{M}_2$ and  $\mathcal{P}_2=\mathcal{M}_2 \setminus \mathcal{M}_1$ be the indices of the messages available only at senders $\mathcal{S}_1$ and $\mathcal{S}_2$ respectively. Let $\mathcal{P}_{\{1,2\}}=\mathcal{M}_1  \cap \mathcal{M}_2$ be the indices of the  messages in common at the senders. Let $m_{\mathcal{A}}=|\mathcal{P}_\mathcal{A}|$ for all non-empty $\mathcal{A} \subseteq \{1,2\}$. We represent a singleton set without a $\{\}$. For example, $\{1\}$ is represented as $1$. Hence, $m_{\{1\}}$ is written as $m_1$. Consider the TGICP $\mathcal{I}$ as an SGICP with the fitting matrix ${\bf{F}}_x$. We partition ${\bf{F}}_x$ based on $(i)$ the message sets ${\bf{x}}_{\mathcal{P}_{\mathcal{A}}}$, for non-empty $\mathcal{A} \subseteq \{1,2\}$, and $(ii)$ the receivers demanding messages in these sets as shown in (\ref{eqfitpart}).
\begin{equation}
{\bf{F}}_x =
\left(
\begin{array}{c|c|c} 
{\bf{F}}_x^{(1)}  & {\bf{A}}^{(1,2)}_{x_0} & {\bf{A}}^{(1,\{1,2\})}_{x_0}\\
\hline
{\bf{A}}^{(2,1)}_{x_0} &  {\bf{F}}_x^{(2)} & {\bf{A}}_{x_0}^{(2,\{1,2\})} \\
\hline
{\bf{A}}^{(\{1,2\},1)}_{x_0} & {\bf{A}}^{(\{1,2\},2)}_{x_0} & {\bf{F}}_x^{(\{1,2\})}
\end{array}
\right).
\label{eqfitpart}
\end{equation}
The first $m_1$ columns correspond to messages ${\bf{x}}_{\mathcal{P}_1}$. The $m_2$ consecutive columns from the $(m_1+1)$th column correspond to messages ${\bf{x}}_{\mathcal{P}_2}$. The remaining columns correspond to messages ${\bf{x}}_{\mathcal{P}_{\{1,2\}}}$.
The matrices ${\bf{A}}^{(\mathcal{A}_1,\mathcal{A}_2)}_{x_0}$ only consist of $x$'s and $0$'s, where $\mathcal{A}_1$ and $\mathcal{A}_2$ are non-empty subsets of $\{1,2\}$ such that $\mathcal{A}_1 \neq \mathcal{A}_2$. These  matrices have the subscript $x_0$ to indicate that they consist of only $x$'s and $0$'s. Let the fitting matrix ${\bf{F}}_x^{(\mathcal{A})}$ of size $n_{\mathcal{A}} \times m_{\mathcal{A}}$, for any non-empty subset $\mathcal{A} \subseteq \{1,2\}$ correspond to an SGICP denoted by  $\mathcal{I}_{\mathcal{A}}$. Note that the side information of all the receivers in  $\mathcal{I}_{\mathcal{A}}$ is present only in ${\bf{x}}_{\mathcal{P}_{\mathcal{A}}}$, for any non-empty subset $\mathcal{A} \subseteq \{1,2\}$. In the remaining part of this paper, whenever we refer to `the sub-problems of the TGICP $\mathcal{I}$', we refer to the SGICPs $\mathcal{I}_1$,  $\mathcal{I}_2$, and $\mathcal{I}_{\{1,2\}}$, unless otherwise stated.

\par For the given TGICP $\mathcal{I}$, if some message in the side information of some receiver belonging to $\mathcal{I}_{\mathcal{A}_1}$ is present in ${\bf{x}}_{\mathcal{P}_{\mathcal{A}_2}}$,  for non-empty subsets $\mathcal{A}_i \subseteq \{1,2\}$, $i \in \{1,2\}$, $\mathcal{A}_1 \neq \mathcal{A}_2$, we say that there is an interaction from $\mathcal{I}_{\mathcal{A}_1}$ to $\mathcal{I}_{\mathcal{A}_2}$ and is denoted by $\mathcal{I}_{\mathcal{A}_1} \rightarrow \mathcal{I}_{\mathcal{A}_2}$. We say that this interaction $\mathcal{I}_{\mathcal{A}_1} \rightarrow \mathcal{I}_{\mathcal{A}_2}$ is fully-participated if ${\bf{A}}_{x_0}^{(\mathcal{A}_1,\mathcal{A}_2)}={\bf{X}}$ (${\bf{X}}$ is a matrix with all unknowns). That is, every receiver in $\mathcal{I}_{\mathcal{A}_1}$ has ${\bf{x}}_{\mathcal{P}_{\mathcal{A}_2}}$ in its side information. If  ${\bf{A}}_{x_0}^{(\mathcal{A}_1,\mathcal{A}_2)} \neq {\bf{X}}$, then the interaction $\mathcal{I}_{\mathcal{A}_1} \rightarrow \mathcal{I}_{\mathcal{A}_2}$ is said to be partially-participated. 
For the given TGICP $\mathcal{I}$, consider the digraph $\mathcal{H}$ with  $\mathcal{V}(\mathcal{H}) = \{1,2,\{1,2\}\}$ and $\mathcal{E}(\mathcal{H})=\{(\mathcal{A}_1,\mathcal{A}_2) |~ \exists ~ \mathcal{I}_{\mathcal{A}_1} \rightarrow \mathcal{I}_{\mathcal{A}_2}, \text{for non-empty } \mathcal{A}_i \subseteq  \{1,2\}, \mathcal{A}_1 \neq \mathcal{A}_2, i \in \{1,2\} \}$. The graph $\mathcal{H}$ is called the interaction digraph of the TGICP $\mathcal{I}$. There are totally 64 possibilities for the graph $\mathcal{H}$ as shown in Fig. \ref{fig1}. The numbers written below each interaction digraph in the figure is the subscript used to denote the specific interaction digraph. Note that the set of possible interaction digraphs given for the TUICP in \cite{CTLO} is same as that given in Fig. \ref{fig1}, but now they describe the interactions between the three constituent SGICPs of the TGICP. For example, the first interaction digraph with no interactions between the sub-problems is denoted as $\mathcal{H}_1$. Note that a pair of interactions given by $\mathcal{I}_{\mathcal{A}_1} \rightarrow \mathcal{I}_{\mathcal{A}_2}$ and $\mathcal{I}_{\mathcal{A}_2} \rightarrow \mathcal{I}_{\mathcal{A}_1}$ is denoted by lines having two-sided arrows in the interaction digraph, for any non-empty subsets $\mathcal{A}_i \subseteq  \{1,2\}, i \in \{1,2\}$. From Fig. \ref{fig1}, note that the TGICP is broadly classified into two cases. If the interaction digraph of the TGICP is acyclic, then the TGICP belongs to Case I, otherwise it belongs to Case II. Case II is further classified into five sub-cases based on the interaction digraph. 

\begin{figure*}[!htbp]
	\begin{center}
		\includegraphics[width=43pc]{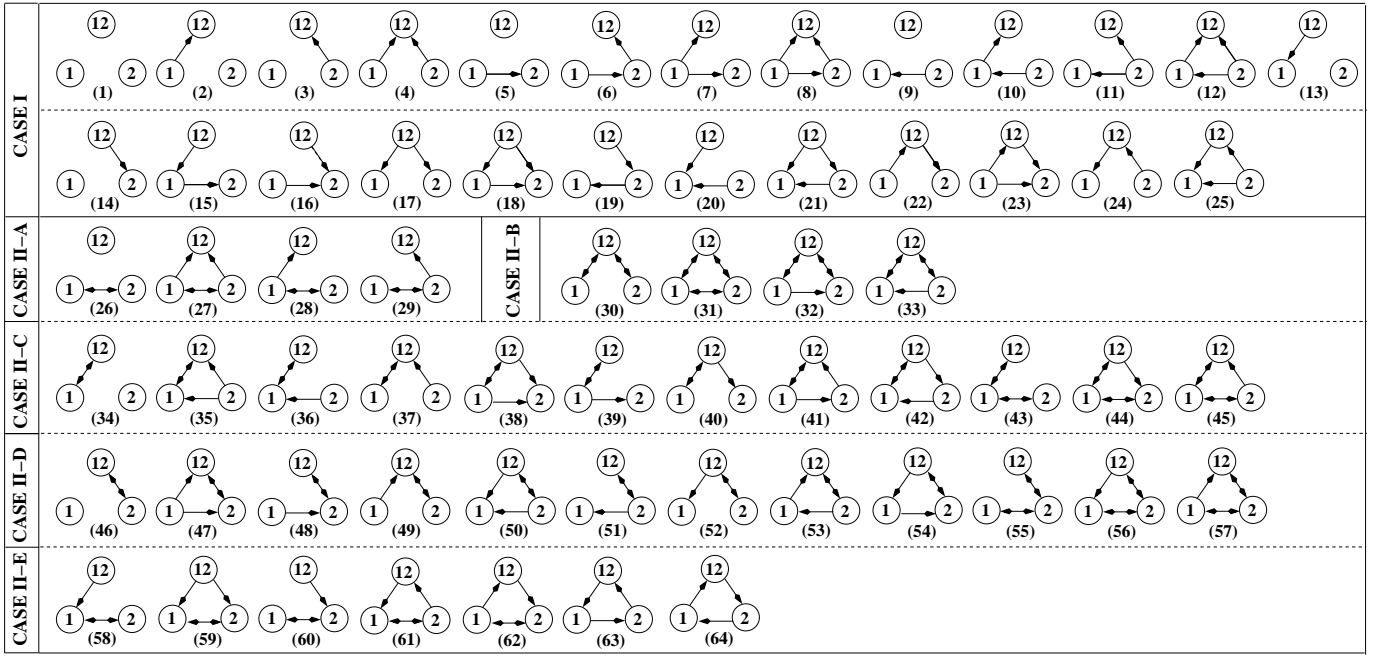}
		\caption{Enumeration of all possible interactions between the sub-problems $\mathcal{I}_1$, $\mathcal{I}_2$, and $\mathcal{I}_{\{1,2\}}$, denoted by the interaction digraph $\mathcal{H}$.}
		\label{fig1}
	\end{center}
	\hrule
\end{figure*}

The following notations are related to the construction of a two-sender index code from single-sender index codes. Let ${\bf{c}}^{(1)}$ and ${\bf{c}}^{(2)}$ be  two codewords of length $l_1$ and $l_2$ respectively. The symbol-wise addition of ${\bf{c}}^{(1)}$ and ${\bf{c}}^{(2)}$ after zero-padding the shorter codeword at the least significant positions to match the length of the longer codeword is denoted by ${\bf{c}}^{(1)} + {\bf{c}}^{(2)}$. For example, assuming the code to be over $\mathbb{F}_2$, if ${\bf{c}}^{(1)}=1010$, and ${\bf{c}}^{(2)}=110$, then ${\bf{c}}^{(1)} + {\bf{c}}^{(2)} = 0110$.  The vector obtained by picking the code-symbols from symbol numbered $a$ to symbol numbered $b$, starting from the most significant position of the codeword ${\bf{c}}^{(i)}$, with $a,b \in [l_i]$, is denoted by ${\bf{c}}^{(i)}[a:b]$. For example ${\bf{c}}^{(1)}[2:4]=010$.

%\begin{figure}[!htbp]
%\centering
%\begin{center}
%\includegraphics[width=19pc]{examp}
%\caption{Example to illustrate the interaction digraph $\mathcal{J}$ and the sub-digraphs of a given side-information digraph $\mathcal{D}$}
%\label{examp1}
%\end{center}
%%\hrule
%\end{figure}

\begin{exmp}
\par Consider the TGICP with the partitioned fitting matrix ${\bf{F}}_x$ given below. From the partition we observe that $\mathcal{P}_1=\{1,2,3\}$, $\mathcal{P}_2=\{4,5\}$, and $\mathcal{P}_{\{1,2\}}=\{6,7,8\}$. Note from the partition that $m_1=3$, $n_1=4$, $m_2=2$, $n_2=2$, $m_{\{1,2\}}=3$, and $n_{\{1,2\}}=3$. Thus $m=m_1+m_2+m_{\{1,2\}}=8$ and $n=n_1+n_2+n_{\{1,2\}}=9$. From the second row of the fitting matrix, we see that the side information of the second receiver is given by ${\bf{x}}_{\rchi_2}$, where $\rchi_2=\{2,4\}$, and it demands ${\bf{x}}_{f(2)}$, where $f(2)=3$. Similarly, the side information and demands of all other receivers can be obtained from the fitting matrix. As receivers $1$ and $4$ demand message ${\bf{x}}_1$, the sub-problem $\mathcal{I}_1$ and hence the two-sender ICP are groupcast ICPs. Note that only the following interactions exist: $\mathcal{I}_1 \rightarrow \mathcal{I}_2$, $\mathcal{I}_{\{1,2\}} \rightarrow \mathcal{I}_{2}$. From the listing of the interaction digraphs in Figure \ref{fig1}, we see that the interaction digraph related to the TGICP is $\mathcal{H}_{16}$. Observe that the interaction 
$\mathcal{I}_{\{1,2\}} \rightarrow \mathcal{I}_{2}$ is fully-participated, and $\mathcal{I}_{1} \rightarrow \mathcal{I}_{2}$ is partially-participated. 
 	\[
 	{\bf{F}}_x=
 	\left(
 	\begin{array}{ccc|cc|ccc} 
 	1 & x & 0 & x & x & 0 & 0 & 0 \\
 	0 & x & 1 & x & 0 & 0 & 0 & 0 \\
 	x & 1 & 0 & x & x & 0 & 0 & 0 \\
 	1 & 0 & x & x & x & 0 & 0 & 0 \\
 	\hline
 	0 & 0 & 0 & 1 & x & 0 & 0 & 0 \\
 	0 & 0 & 0 & x & 1 & 0 & 0 & 0 \\
 	\hline
 	0 & 0 & 0 & x & x & 1 & 0 & x \\
 	0 & 0 & 0 & x & x & x & 1 & 0 \\
 	0 & 0 & 0 & x & x & 0 & x & 1 \\
 	\end{array}
 	\right).
 	\]
\end{exmp}
%%%%%%%%%%%%%%%%%%%%%%%%%%%%%%%%%%%%%%%%%%%%%%%%%%%%%%%%%%%%%%%%%%%%
\section{Optimal scalar linear codes for the TGICPs belonging to Cases I and II-A}

In this section, we construct optimal scalar linear codes for TGICPs belonging to Cases I and II-A, with any type of interactions among the sub-problems. The following lemma is required to prove the main result of this section (Theorem \ref{thmcase1}). The proof follows on similar lines as that of Lemma 4.2 in \cite{DSC1}. We provide the proof for completeness.
 
\begin{lem}
\label{lmSCC}
Consider the SGICP with the fitting matrix ${\bf{F}}_x$  being a block upper triangular matrix (all the entries below the block matrices $\{{\bf{F}}_x^{(i)}\}_{i \in [b]}$ are zeros) as given in (\ref{eqbup}). The matrices  ${\bf{B}}^{(i,j)}_{x_0}$, for all  $i,j \in [b]$ with  $i < j$, consist of only $x$'s and $0$'s. Then, $\text{mrk}_{q}({\bf{F}}_{x}) = \sum_{s=1}^{b} \text{mrk}_{q}({\bf{F}}^{(s)}_{x}).$
\begin{equation}
{\bf{F}}_x =
\left(
\begin{array}{c|c|c|c|c} 
{\bf{F}}_x^{(1)}  & {\bf{B}}^{(1,2)}_{x_0} & \cdots & {\bf{B}}_{x_0}^{(1,b-1)} & {\bf{B}}^{(1,b)}_{x_0}\\
\hline
{\bf{0}} &  {\bf{F}}_x^{(2)} & \ddots & {\bf{B}}_{x_0}^{(2,b-1)} & {\bf{B}}_{x_0}^{(2,b)} \\
\hline
\vdots & \ddots & \ddots & \ddots & \vdots \\
\hline
{\bf{0}} & {\bf{0}} & \ddots & {\bf{F}}^{(b-1)}_{x} & {\bf{B}}_{x_0}^{(b-1,b)} \\
\hline
{\bf{0}} & {\bf{0}} & \cdots & {\bf{0}} & {\bf{F}}_{x}^{(b)} 
\end{array}
\right).
\label{eqbup}
\end{equation}
\end{lem}
\begin{proof}
	If ${\bf{F}}^{(i)} \approx {\bf{F}}^{(i)}_x$, $i \in [b]$, then the matrix ${\bf{F}}$ with its $i$th diagonal block matrix given by ${\bf{F}}^{(i)}$, and non-diagonal block matrices being all-zero matrices, completes ${\bf{F}}_x$. Thus 
	$\text{rk}_q({\bf{F}})=\sum_{s=1}^{b} \text{rk}_{q}({\bf{F}}^{(s)})$, which implies that $\text{mrk}_q({\bf{F}}_x) \leq \sum_{s=1}^{b} \text{mrk}_{q}({\bf{F}}^{(s)}_x)$, which lower bounds  $\text{mrk}_q({\bf{F}}_x)$. We now obtain a matching lower bound. Let ${\bf{\tilde{F}}} \approx {\bf{F}}_x$, and $\{{\bf{\tilde{F}}}^{(i)}\}_{i \in [b]}$, be the diagonal block matrices  of ${\bf{\tilde{F}}}$ such that ${\bf{\tilde{F}}}^{(i)} \approx {\bf{F}}^{(i)}_x$. Then,  $\text{rk}_q({\bf{\tilde{F}}}) \geq \sum_{s=1}^{b} \text{rk}_{q}({\bf{\tilde{F}}}^{(s)}) \geq \sum_{s=1}^{b} \text{mrk}_{q}({\bf{{F}}}^{(s)}_x) $. Hence, we have $\text{mrk}_q({\bf{F}}_x) \geq \sum_{s=1}^{b} \text{mrk}_{q}({\bf{F}}^{(s)}_x)$, which is a matching lower bound.       
\end{proof}

We require the notion of topological ordering and a related lemma to establish our result for Case I. 

\begin{defn}[Topological ordering, \cite{DEK}]
	A topological ordering of a digraph $\mathcal{D}$ is a labelling of its vertices using the numbers in $\{1,2,\cdots,|\mathcal{V}(\mathcal{D})|\}$, such that for every edge $(u,v) \in \mathcal{E}(\mathcal{D})$,  $u < v$, where $u,v \in \{1,2,\cdots,|\mathcal{V}(\mathcal{D})|\}$.
\end{defn}

\begin{lem}[\cite{DEK}]
	\label{lemtopord}
	Any acyclic digraph has at least one topological ordering. 
\end{lem}

Next, we state the main result of this section.
\begin{thm}
\label{thmcase1}
For any TGICP $\mathcal{I}$ belonging to either Case I or Case II-A, having any type of interactions (fully-participated or partially-participated) among its sub-problems $\mathcal{I}_{1},\mathcal{I}_{2}$, and $\mathcal{I}_{\{1,2\}}$, we have,
\begin{gather*}
l^*_{q}(\mathcal{I})=\text{mrk}_{q}({\bf{F}}_x^{(1)})+\text{mrk}_{q}({\bf{F}}_x^{(2)})+\text{mrk}_{q}({\bf{F}}_x^{(\{1,2\})}).
\end{gather*}
\end{thm}
\begin{proof} We first prove the  theorem for Case I. The interaction digraph of any TGICP belonging to Case I is acyclic. According to Lemma \ref{lemtopord}, the interaction digraph has a topological ordering. Hence, by a permutation of the rows of ${\bf{F}}_x$ and/or by a permutation of the columns of ${\bf{F}}_x$, one can obtain another fitting matrix ${\bf{F}}'_x$ which is block upper triangular similar to the one given in (\ref{eqbup}). This amounts to relabelling of receivers (if rows are permuted) and/or  messages (if columns are permuted). Considering the TGICP as an SGICP with the fitting matrix ${\bf{F}}'_x$ and using Lemma \ref{lmSCC}, we obtain $l^*_q(\mathcal{I}) \geq \text{mrk}_{q}({\bf{F}}_{x}) =\sum_{s=1}^{3} \text{mrk}_{q}({\bf{F}}^{(s)}_{x})$, which is a lower bound on $l^*_q(\mathcal{I})$. By transmitting the optimal scalar linear index codes for the SGICPs $\mathcal{I}_1$, $\mathcal{I}_2$, and $\mathcal{I}_{\{1,2\}}$ from appropriate senders, we have $l^*_q(\mathcal{I}) \leq \sum_{s=1}^{3} \text{mrk}_{q}({\bf{F}}^{(s)}_{x})$. This gives a matching upper bound on $l^*_q(\mathcal{I})$. 
	
We now prove the theorem for Case II-A. The proof follows on similar lines as that of Theorem $2$ in \cite{CVBSR1}. We use interference alignment techniques to prove the result. Without loss of generality, any optimal  scalar linear code  can be written as ${\bf{G}}{\bf{x}}_{[m]}$, where ${\bf{G}}$ is given in  (\ref{Geq}).
Let ${\bf{G}}^{(i)} \in \mathbb{F}_q^{l'_i \times |\mathcal{P}_i|},$ $i \in \{1,2\}$, and ${\bf{\tilde{G}}}^{(j)} \in \mathbb{F}_q^{l'_j \times |\mathcal{P}_{\{1,2\}}|},$ $j \in \{1,2,3\}$. We assume that the matrices $({\bf{G}}^{(1)} | {\bf{\tilde{G}}}^{1}), ({\bf{G}}^{(2)} | {\bf{\tilde{G}}}^{(2)}),$ and ${\bf{\tilde{G}}}^{(3)}$ are full-rank, which is required for the optimality of the code. The columns of ${\bf{G}}$ are the precoding vectors for the messages, according to the interference alignment perspective \cite{jafar}. Note that the number of received signal dimensions (at any receiver) is same as the number of rows of ${\bf{G}}$.
\begin{equation}
\label{Geq}
{\bf{G}}=
\left(
\begin{array}{c|c|c} 
{\bf{G}}^{(1)} & {\bf{0}} & {\bf{\tilde{G}}}^{(1)} \\
\hline
{\bf{0}} & {\bf{G}}^{(2)} & {\bf{\tilde{G}}}^{(2)} \\
\hline
{\bf{0}} & {\bf{0}} & {\bf{\tilde{G}}}^{(3)} \\
\end{array}
\right).
\end{equation}

The upper bound given in the proof of the theorem for Case I is  also an upper bound for Case II-A. We now prove that it is also a matching lower bound for Case II-A. The receivers in $\mathcal{I}_{\{1,2\}}$ do not have any side information in ${\bf{x}}_{\mathcal{P}_1} \cup  {\bf{x}}_{\mathcal{P}_2}$. The precoding vectors for the messages in ${\bf{x}}_{\mathcal{P}_1} \cup {\bf{x}}_{\mathcal{P}_2}$ must be independent of the precoding vectors for the messages in ${\bf{x}}_{\mathcal{P}_{\{1,2\}}}$. Otherwise, one or more receivers in $\mathcal{I}_{\{1,2\}}$ can not cancel the interference caused by the precoding vectors of one or more messages in ${\bf{x}}_{\mathcal{P}_1} \cup {\bf{x}}_{\mathcal{P}_2}$. Thus, a minimum of $\text{mrk}({\bf{F}}_x^{(\{1,2\})})$ independent vectors are required to satisfy all the receivers in $\mathcal{I}_{\{1,2\}}$, as it is the optimal scalar linear code length for the SGICP $\mathcal{I}_{\{1,2\}}$. The precoding vectors of the messages in ${\bf{x}}_{\mathcal{P}_{1}}$ are independent of the precoding vectors of those  in ${\bf{x}}_{\mathcal{P}_{2}}$ (from the structure of ${\bf{G}}$). Hence, a  minimum of $\text{mrk}({\bf{F}}_x^{(1)})+\text{mrk}({\bf{F}}_x^{(2)})$ independent vectors are required to satisfy all the receivers in $\mathcal{I}_{1}$ and $\mathcal{I}_{2}$. The total number of dimensions used for precoding all the messages is same as the rank of ${\bf{G}}$, which thus satisfies $l'_1 + l'_2 +l'_3 \geq \text{mrk}({\bf{F}}_x^{(1)})+\text{mrk}({\bf{F}}_x^{(2)})+\text{mrk}({\bf{F}}_x^{(\{1,2\})}) $. 
\end{proof}

\section{Jointly Extended SGICPs}
\par In this section, we introduce the notion of joint extensions of any finite number of SGICPs, and identify a special class for which optimal scalar linear codes can be constructed using those of the sub-problems. We also show that two classes of  rank-invariant extensions of any SGICP presented in \cite{PK} are special cases of the class of joint extensions identified in this section. We require the following two definitions from \cite{PK}.

\begin{defn}[Extension of an SGICP, \cite{PK}]
An SGICP $\mathcal{I}_{e}$ with fitting matrix ${\bf{F}}_{x}^{e}$ is called an extended SGICP of another SGICP $\mathcal{I}$ (or  an extension of $\mathcal{I}$) with fitting matrix ${\bf{F}}_x$, if ${\bf{F}}_{x}^{e}$ contains ${\bf{F}}_x$ as a submatrix.  
\end{defn}

\begin{defn}[Rank-Invariant Extension, \cite{PK}]
An extension $\mathcal{I}_{e}$ of the SGICP $\mathcal{I}$ is called a rank-invariant extension of $\mathcal{I}$, if the optimal  scalar linear code lengths of both the SGICPs are equal.
\end{defn}

\par We now define a joint extension of a finite number of SGICPs which generalizes the notion of extension of a single SGICP.

\begin{defn}
Consider $u$ SGICPs where the $i$th SGICP $\mathcal{I}_i$ is described using the fitting matrix ${\bf{F}}_x^{(i)}$, $i \in [u]$. An SGICP $\mathcal{I}_{E}$ with the fitting matrix  ${\bf{F}}^{E}_{x}$ is called a jointly extended SGICP (or a joint extension) of $u$ SGICPs $\{\mathcal{I}_i\}_{i \in [u]}$, if 
${\bf{F}}^{E}_{x}$ consists of all ${\bf{F}}_x^{(i)}$'s, $i \in [u]$, as disjoint submatrices. 
\end{defn}

\par Note that  $\mathcal{I}_{E}$ is also an extension of $\mathcal{I}_{i}$, for all $i \in [u]$. Lemmas $1$ and $2$ in \cite{PK} are used to derive the results in this section and are stated below.

\begin{lem}[Lemma 1, \cite{PK}]
	For the SGICP $\mathcal{I}$ with $n \times m$ fitting matrix ${\bf{F}}_x$, a matrix ${\bf{G}} \in \mathbb{F}^{r \times m}_q$ is an encoding matrix (for some linear index code) iff there exists a matrix ${\bf{D}} \in \mathbb{F}^{n \times r}_q$ such that ${\bf{D}}{\bf{G}}$ completes ${\bf{F}}_x$, i.e. ${\bf{D}}{\bf{G}} \approx  {\bf{F}}_x$.
	\label{dgmatrix}
\end{lem}

\begin{lem}[Lemma 2, \cite{PK}]
	If the fitting matrix ${\bf{F}}_x$ of the SGICP $\mathcal{I}$ is a submatrix of the fitting matrix ${\bf{F}}^{e}_{x}$ of the SGICP $\mathcal{I}_{e}$, then $\text{mrk}_{q}({\bf{F}}^{e}_{x}) \geq \text{mrk}_{q}({\bf{F}}_x)$.
	\label{minrk}
\end{lem}

\par We now state and prove the main result of this section. The proof is based on the approach used to prove Theorem 2, \cite{LRM}. 

\begin{thm}
	\label{th1}
Consider the jointly extended SGICP $\mathcal{I}_E$ with the $(\sum_{i=1} ^{u}n_i) \times (\sum_{i=1} ^{u}m_i)$ fitting matrix ${\bf{F}}_x^{E}$ given in (\ref{nJEICP}), where the $i$th sub-problem $\mathcal{I}_i$ has the $n_i \times m_i$ fitting matrix ${\bf{F}}_x^{(i)}$, $i \in [u]$. Let  ${\bf{F}}^{(i)} \approx {\bf{F}}_x^{(i)}$, with the first $r_i=\text{rk}_q({\bf{F}}^{(i)})$ rows of ${\bf{F}}^{(i)}$ spanning $\langle {\bf{F}}^{(i)} \rangle$, $\forall i \in [u]$, and $r_i$ not necessarily equal to $\text{mrk}_{q}({\bf{F}}^{(i)}_x)$. Also, let 
$r_1 \geq r_2 \geq ... \geq r_{u-1} \geq r_u$. Let ${\bf{B}}_{x_0}^{(i,j)}$ be any $n_i \times m_j$ matrix such that ${\bf{B}}^{(i,j)} \approx {\bf{B}}_{x_0}^{(i,j)}$, $\forall i,j \in [u], i \neq j$, where ${\bf{B}}^{(i,j)} = \left( \begin{array}{c}  {\bf{\hat{F}}}^{(i,j)} \\
\hline
{\bf{P}}^{(i)}{\bf{\hat{F}}}^{(i,j)}
\end{array} \right)$, with ${\bf{\hat{F}}}^{(i,j)}$ given in (\ref{Fcap}), and ${\bf{P}}^{(i)}$ being an $(n_i-r_i) \times r_i$ matrix such that the last $n_i-r_i$ rows of ${\bf{F}}^{(i)}$ are given by ${\bf{P}}^{(i)}{\bf{F}}^{(i)}_{[r_i]}$.
\begin{equation}
{\bf{\hat{F}}}^{(i,j)} =
\begin{cases}
{\bf{F}}^{(j)}_{[r_i]} & if \ r_j \geq r_i,\\
\left( \begin{array}{c} {\bf{F}}^{(j)}_{[r_j]}\\
\hline
{\bf{0}}_{(r_i-r_j) \times m_j}\\
\end{array}
\right) & otherwise.
\end{cases}
\label{Fcap}
\end{equation}
An encoding matrix for the SGICP $\mathcal{I}_{E}$ is given by ${\bf{G}}_{E}=\begin{pmatrix}
{\bf{\hat{F}}}^{(1,1)}|{\bf{\hat{F}}}^{(1,2)}|...|...|{\bf{\hat{F}}}^{(1,u)}\end{pmatrix}$. Moreover, ${\bf{G}}_{E}{\bf{x}}_{[m]}$ is an optimal scalar linear code if $r_1=\max\{\text{mrk}_{q}({\bf{F}}^{(i)}_x),i \in [u]\}$.
	\begin{equation}
	{\bf{F}}^{E}_{x} =		%\begin{pmatrix}
	\left(
	\begin{array}{c|c|c|c|c}
	{\bf{F}}_x^{(1)}  & {\bf{B}}_{x_0}^{(1,2)}  & ... & ...       & {\bf{B}}_{x_0}^{(1,u)}  \\
	\hline
	{\bf{B}}_{x_0}^{(2,1)}          & {\bf{F}}_x^{(2)} & {\bf{B}}_{x_0}^{(2,3)} & \ldots      & \vdots  \\
	\hline
	\vdots       &  {\bf{B}}_{x_0}^{(3,2)}  & \ddots & \ldots      & \vdots  \\
	\hline
	\vdots      & \vdots & \vdots & {\bf{F}}_x^{(u-1)} &  {\bf{B}}_{x_0}^{(u-1,u)} \\
	\hline
	{\bf{B}}_{x_0}^{(u,1)}    & \ldots & \ldots &    {\bf{B}}_{x_0}^{(u,u-1)}      & {\bf{F}}_x^{(u)} \\
	\end{array}
	\right).
		\label{nJEICP}
	%\end{pmatrix}.
	\end{equation}	
\end{thm}
\begin{proof}
	\par We first show that there exists a matrix ${\bf{D}}_{E}$ such that ${\bf{D}}_{E}{\bf{G}}_{E} \approx {\bf{F}}^{E}_{x}$, so that ${\bf{G}}_{E}$ is an encoding matrix according to Lemma \ref{dgmatrix}. Consider the $n_i \times r_1$ matrix 
	\begin{center}
		${\bf{D}}^{(i)}=\left( \begin{array}{c|c} {\bf{I}}_{r_i \times r_i} & {\bf{0}}_{r_i \times (r_1-r_i)}\\ 
		\hline
		{\bf{P}}^{(i)}  & {\bf{0}}_{(n_i-r_i) \times (r_1-r_i)} \end{array} \right)$.
	\end{center}
	\par It can be easily verified using multiplication of block matrices that, ${\bf{D}}^{(i)}{\bf{\hat{F}}}^{(1,j)}={\bf{B}}^{(i,j)}$, and ${\bf{D}}^{(i)}{\bf{\hat{F}}}^{(1,i)}={\bf{F}}^{(i)}$, $\forall i \neq j, i, j \in [u]$. Let ${\bf{D}}_{E} = 
	\begin{pmatrix}
	({\bf{D}}^{(1)})^T | ({\bf{D}}^{(2)})^T  |...|({\bf{D}}^{(u)})^T
	\end{pmatrix}^T$. Then the product ${\bf{D}}_{E}{\bf{G}}_{E}$ is as given in (\ref{FE}).
	\begin{equation}  
	\begin{pmatrix} 
	{\bf{D}}^{(1)}{\bf{\hat{F}}}^{(1,1)} & {\bf{D}}^{(1)}{\bf{\hat{F}}}^{(1,2)} & ...    & {\bf{D}}^{(1)}{\bf{\hat{F}}}^{(1,u)} \\
	{\bf{D}}^{(2)}{\bf{\hat{F}}}^{(1,1)} & {\bf{D}}^{(2)}{\bf{\hat{F}}}^{(1,2)} & ...    & {\bf{D}}^{(2)}{\bf{\hat{F}}}^{(1,u)} \\
	\vdots           & \vdots           & ...    & \vdots \\ 
	{\bf{D}}^{(u)}{\bf{\hat{F}}}^{(1,1)} & {\bf{D}}^{(u)}{\bf{\hat{F}}}^{(1,2)} & ...    & {\bf{D}}^{(u)}{\bf{\hat{F}}}^{(1,u)} \\
	\end{pmatrix}
	\approx {\bf{F}}^{E}_{x}.
	\label{FE}
	\end{equation}
	\par Hence, the length of the code obtained from ${\bf{G}}_{E}$ is $r_1$. From Lemma \ref{minrk}, $\text{mrk}_{q}({\bf{F}}^{E}_{x}) \geq \max\{\text{mrk}_{q}({\bf{F}}_x^{(i)}),i \in [u]\}$. Hence if $r_{1}=\max\{\text{mrk}_{q}({\bf{F}}_x^{(i)}),i \in [u]\}$, the code ${\bf{G}}_{E}{\bf{x}}_{[m]}$ is scalar linear optimal. 
\end{proof}

The following corollary shows that Lemma $7$ in \cite{CVBSR1} can be obtained as a special case of Theorem \ref{th1} by making a minor modification in its proof.

\begin{cor}
	Consider the SGICP $\mathcal{I}_E$ with the fitting matrix ${\bf{F}}^{E}_{x}$ given in Theorem \ref{th1} with ${\bf{B}}_{x_0}^{(i,j)}={\bf{X}}$, for all $i \neq j, i,j \in [u]$. Let ${\bf{c}}^{(i)}$ be any scalar linear code (not necessarily optimal) for the SGICP $\mathcal{I}_i$ with code length $r_i$. Then the fact that ${\bf{c}}^{(1)}+\cdots+{\bf{c}}^{(u)}$ is a scalar linear code for $\mathcal{I}_E$ follows from Theorem \ref{th1}. It is an optimal scalar linear code if $r_{max}=\max\{r_i,i \in [u]\}=\max\{\text{mrk}_{q}({\bf{F}}^{(i)}_x),i \in [u]\}$.
	\label{corth1}
\end{cor}
\begin{proof}
	Let ${\bf{\tilde{G}}}^{(i)}$ be an $r_i \times m_i$ encoding matrix for the code ${\bf{c}}^{(i)}$, $i \in [u]$.  From Lemma \ref{dgmatrix}, there exists an $n_i \times r_i$ matrix ${\bf{\tilde{D}}}^{(i)}$ such that ${\bf{\tilde{D}}}^{(i)}{\bf{\tilde{G}}}^{(i)}={\bf{\tilde{F}}}^{(i)} \approx {\bf{F}}_x^{(i)}$, $i \in [u]$. Without loss of generality, we can assume that the first $r'_i=\text{rk}_q({\bf{\tilde{F}}}^{(i)})$  rows of ${\bf{\tilde{F}}}^{(i)}$ are linearly independent, for all $i \in [u]$. Otherwise, applying a suitable permutation on the rows of 
	${\bf{\tilde{F}}}^{(i)}$ given by an $n_i \times n_i$ permutation matrix ${\bf{\tilde{P}}}^{(i)}$, we obtain  ${\bf{\tilde{P}}}^{(i)}{\bf{\tilde{F}}}^{(i)}={\bf{F}}^{(i)}$ such that the first $r'_i$ rows of ${\bf{F}}^{(i)}$ are linearly independent, where $i \in [u]$. The same permutation must be applied on the rows of ${\bf{F}}_x^{E}$ containing ${\bf{F}}^{(i)}_x$, for all $i \in [u]$. This does not change the SGICP as this amounts to relabelling the receivers. Similarly as ${\bf{B}}_{x_0}^{(i,j)}={\bf{X}}$, for all $i \neq j, i,j \in [u]$, without loss of generality, we can assume that $r'_1 \geq r'_2 \geq \cdots \geq r'_u$. Otherwise, another permutation on the sets of columns of ${\bf{F}}_x^{E}$ containing ${\bf{F}}^{(i)}_x$, for all $i \in [u]$, can be applied to obtain the desired order. This amounts to relabelling the messages and the SGICP $\mathcal{I}_E$ does not change. Hence, all the conditions given in Theorem \ref{th1} are satisfied. Thus, we have the $r'_1 \times m_i$ matrix ${\bf{\hat{F}}}^{(1,i)}={\bf{\tilde{F}}}^{(i)}_{[r'_i]}=\left( \begin{array}{c|c} {\bf{I}}_{r'_i \times r'_i} & {\bf{0}}\\ 
	\hline
	{\bf{0}}_{(r'_1-r'_i) \times r'_i}  & {\bf{0}} \end{array} \right) {\bf{\tilde{D}}}^{(i)} {\bf{\tilde{G}}}^{(i)}$, $i \in [u]$. Note that ${\bf{\tilde{D}}}^{(i)} {\bf{\tilde{G}}}^{(i)}={\bf{\hat{D}}}^{(i)} {\bf{\hat{G}}}^{(i)}$, where  ${\bf{\hat{D}}}^{(i)} = \left( \begin{array}{c|c} {\bf{\tilde{D}}}^{(i)} & {\bf{0}}_{n_i \times (r_{max}-r_i)}
	\end{array} 
	\right)$, and ${\bf{\hat{G}}}^{(i)} = \left( \begin{array}{c} {\bf{\tilde{G}}}^{(i)}\\
	\hline
	{\bf{0}}_{(r_{max}-r_i) \times m_i} \end{array} \right)$. The $n_i \times r'_1$ matrix ${\bf{D}}^{(i)}$ that satisfies 
	${\bf{D}}^{(i)}{\bf{\hat{F}}}^{(1,i)}={\bf{F}}^{(i)}$ (as given in the proof of Theorem \ref{th1}), is given by $			\left( \begin{array}{c|c} {\bf{I}}_{r'_i \times r'_i} & {\bf{0}}\\ 
	\hline
	{\bf{P}}^{(i)}  & {\bf{0}} \end{array} \right)$, where ${\bf{P}}^{(i)}$ is an $(n_i-r'_i) \times r'_i$ matrix such that the last $n_i-r'_i$ rows of ${\bf{\tilde{F}}}^{(i)}$ are given by ${\bf{P}}^{(i)}{\bf{\tilde{F}}}^{(i)}_{[r'_i]}$, $i \in [u]$. Hence, we have
	\begin{eqnarray}
	{\bf{D}}^{(i)}{\bf{\hat{F}}}^{(1,i)} = {\bf{D}}^{(i)}\left( \begin{array}{c|c} {\bf{I}}_{r'_i \times r'_i} & {\bf{0}}\\ 
	\hline
	{\bf{0}}_{(r'_1-r'_i) \times r'_i}  & {\bf{0}} \end{array} \right) {\bf{\hat{D}}}^{(i)} {\bf{\hat{G}}}^{(i)} = {\bf{F}}^{(i)}.
	\label{eqbig}
	\end{eqnarray}
	Now considering the matrix ${\bf{D}}^{(i)}\left( \begin{array}{c|c} {\bf{I}}_{r'_i \times r'_i} & {\bf{0}}\\ 
		\hline
		{\bf{0}}_{(r'_1-r'_i) \times r'_i}  & {\bf{0}} \end{array} \right) {\bf{\hat{D}}}^{(i)}$ as ${\bf{\bar{D}}}^{(i)}$, and using (\ref{eqbig}), we see that  ${\bf{\bar{D}}}^{(i)}{\bf{\hat{G}}}^{(i)} = {\bf{F}}^{(i)}$, for all $i \in [u]$.
	    Hence, by taking ${\bf{G}}_{E}=\begin{pmatrix}
			{\bf{\hat{G}}}^{(1)}|...|...|{\bf{\hat{G}}}^{(u)}\end{pmatrix}$ and ${\bf{D}}_{E} = 
			\begin{pmatrix}
			({\bf{\bar{D}}}^{(1)})^T | ({\bf{\bar{D}}}^{(2)})^T  |...|({\bf{\bar{D}}}^{(u)})^T
			\end{pmatrix}^T$, we see that ${\bf{D}}_{E}{\bf{G}}_{E} \approx {\bf{F}}_x^{E}$ (as in the proof of Theorem \ref{th1}). Thus, ${\bf{G}}_{E}$ is an encoding matrix  for the SGICP $\mathcal{I}_E$ according to Lemma \ref{dgmatrix}.
			Note that the scalar linear code given by ${\bf{G}}_{E}$ is ${\bf{c}}^{(1)}+\cdots+{\bf{c}}^{(u)}$, and the code length is $r_{max}$.  
\end{proof}

\begin{note}
	Note that there is no restriction (in the proof of Corollary \ref{corth1}) that the first $r_i$ rows of ${\bf{\hat{G}}}^{(i)}$  must complete the first $r_i$ rows of  ${\bf{F}}^{(i)}_x$, $i \in [u]$ (as in Theorem \ref{th1}).
	\label{rem1}
\end{note}

\begin{obs}
Theorem $1$ in \cite{PK} can also be obtained as a special case of Theorem \ref{th1}, with all the $u$ sub-problems being the same and having the $n_i \times m_i$ fitting matrix given by ${\bf{F}}_x$. All the block matrices ${\bf{B}}_{x_0}^{(i,j)}$, $i \neq j, i,j \in [u]$ are equal to ${\bf{F}}_x^x$, which is obtained from ${\bf{F}}_x$ by replacing all the $1$'s with $x$'s. The proof follows on similar lines as Corollary \ref{corth1} and Note \ref{rem1} also holds. It can be  shown on similar lines that Theorem $2$ in \cite{PK} can also be obtained as  a special case of Theorem \ref{th1} (in this work). Hence, the result of Theorem \ref{th1} generalizes two previously known results on rank-invariant extensions established in \cite{PK}.
\label{obs1}
\end{obs}

Note that for the encoding matrix given in Theorem $2$ to be scalar linear optimal, encoding matrices of all the sub-problems need not be optimal. Hence we can obtain optimal scalar linear codes for the jointly extended problem, even when sub-optimal scalar linear codes are available for some of the sub-problems,  as long as $r_1=\max\{\text{mrk}_{q}({\bf{F}}^{(i)}_x),i \in [u]\}$. 

\par We now illustrate the completion of ${\bf{F}}_x^{E}$ given in the proof of Theorem \ref{th1} for the case with $u=2$, in (\ref{F2}). This is the case that is used in the next section, to construct scalar linear codes for some classes of the TGICP. Note that the rows from the second (or third, or fourth) row of block matrices are in the span of the rows in the first row of block matrices. 
\begin{equation}
 \left(
    \begin{array}{c|c}
      {\bf{F}}^{(1)}_{[r_1]}  &  \left(\begin{array}{c} {\bf{F}}^{(2)}_{[r_2]} \\
         \hline
          {\bf{0}}_{(r_1-r_2) \times m_2} \end{array} \right) \\
      \hline
          {\bf{P}}^{(1)}{\bf{F}}^{(1)}_{[r_1]}  &   {\bf{P}}^{(1)}\left(\begin{array}{c}
      {\bf{F}}^{(2)}_{[r_2]} \\
      \hline
       {\bf{0}}_{[(r_1-r_2)]} \end{array} \right) \\ 
    \hline
     {\bf{F}}^{(1)}_{[r_2]}  &  {\bf{F}}^{(2)}_{[r_2]} \\ 
     \hline
    {\bf{P}}^{(2)}{\bf{F}}^{(1)}_{[r_2]} &  {\bf{P}}^{(2)}{\bf{F}}^{(2)}_{[r_2]}
    \end{array} 
    \right) \approx {\bf{F}}_x^E.
    \label{F2}
\end{equation}    

We now illustrate the code construction in Theorem \ref{th1}.

\begin{exmp}
Consider the partitioned fitting matrix of a jointly extended SGICP given below.
\[ {\bf{F}}^{E}_{x} = \left(
\begin{array}{ccccc|ccc}
1 & x & x & 0 & 0 & x & 0 & x\\
0 & 1 & x & x & 0 & x & x & 0 \\
0 & 0 & 1 & x & x & 0 & 0 & 0 \\
x & 0 & 0 & 1 & x & x & 0 & x \\
x & x & 0 & 0 & 1 & 0 & x & x \\
\hline
x & 0 & x & 0 & 0 & 1 & 0 & x\\
0 & x & 0 & x & 0 & x & 1 & 0 \\
x & x & x & x & 0 & 0 & x & 1 \\
x & x & x & x & x & x & 0 & 1 \\
\end{array} \right).
\]

From the partition we observe that the SGICPs with fitting matrices ${\bf{F}}^{(1)}_{x}$ and ${\bf{F}}^{(2)}_{x}$ are already solved in \cite{MBR1} and \cite{SUOH}. We provide the completions of these fitting matrices over $\mathbb{F}_2$ below, which follow the conditions given in Theorem \ref{th1}.
 
\[ {\bf{F}}^{(1)} = \left(
\begin{array}{ccccc}
1 & 0 & 1 & 0 & 0 \\
0 & 1 & 0 & 1 & 0 \\
0 & 0 & 1 & 1 & 1 \\
\hline
1 & 0 & 0 & 1 & 1 \\
1 & 1 & 0 & 0 & 1 
\end{array}
\right) ,~ {\bf{F}}^{(2)} = \left(
\begin{array}{ccc}
1 & 0 & 1 \\
1 & 1 & 0 \\
\hline
0 & 1 & 1 \\ 
1 & 0 & 1 \\
\end{array}
\right).
\]
\par Note that $r_1=\text{mrk}_q({\bf{F}}^{(1)}_x)=3$ and  $r_2=\text{mrk}_q({\bf{F}}^{(2)}_x)=2$. It can be  verified that the first $r_i$ rows of ${\bf{F}}^{(i)}$ span $\langle {\bf{F}}^{(i)} \rangle$, $i \in \{1,2\}$. The matrices ${\bf{P}}^{(1)}$ and ${\bf{P}}^{(2)}$ are also given below.
\[ {\bf{P}}^{(1)} = \left(
\begin{array}{ccccc}
1 & 0 & 1 \\
1 & 1 & 1  
\end{array}
\right) ,~ {\bf{P}}^{(2)} = \left(
\begin{array}{ccc}
1 & 1 \\
1 & 0 \\
\end{array}
\right).
\]
\par It can be  verified that ${\bf{B}}^{(ij)} \approx {\bf{B}}_{x_0}^{(ij)}$, $\forall i,j \in [2], i \neq j$, as given in Theorem \ref{th1}, and are given below. 
\[ {\bf{B}}^{(12)} = \left(
\begin{array}{ccccc}
1 & 0 & 1 \\
1 & 1 & 0 \\
0 & 0 & 0 \\
\hline
1 & 0 & 1 \\
0 & 1 & 1  
\end{array}
\right),~ {\bf{B}}^{(21)} = \left(
\begin{array}{ccccccc}
1 & 0 & 1 & 0 & 0\\
0 & 1 & 0 & 1 & 0\\
\hline
1 & 1 & 1 & 1 & 0\\
1 & 0 & 1 & 0 & 0\\ 
\end{array}
\right).
\]
\par  An encoding matrix ${\bf{G}}_{E}$ is obtained as stated in  Theorem \ref{th1}, which is optimal and given below. It can be verified that all receivers' demands are satisfied.
\[ {\bf{G}}_{E} = \left(
\begin{array}{ccccc|ccc}
1 & 0 & 1 & 0 & 0 & 1 & 0 & 1\\
0 & 1 & 0 & 1 & 0 & 1 & 1 & 0\\
0 & 0 & 1 & 1 & 1 & 0 & 0 & 0
\end{array} \right).
\]
The optimal scalar linear codes for the individual problems are ${\bf{c}}^{(1)}=(~{\bf{x}}_1 + {\bf{x}}_3,~ {\bf{x}}_2 + {\bf{x}}_4, ~{\bf{x}}_3 + {\bf{x}}_4 + {\bf{x}}_5 ~)$, and ${\bf{c}}^{(2)}=(~{\bf{x}}_6 + {\bf{x}}_8,~ {\bf{x}}_6 + {\bf{x}}_7~)$
. The optimal scalar linear code for the jointly extended problem given by ${\bf{G}}_{E}$ is ${\bf{c}}=(~{\bf{x}}_1 + {\bf{x}}_3 + {\bf{x}}_6 + {\bf{x}}_8, ~{\bf{x}}_2 + {\bf{x}}_4 + {\bf{x}}_6 + {\bf{x}}_7,~ {\bf{x}}_3 + {\bf{x}}_4 + {\bf{x}}_5 ~)$.
\label{ex1}
\end{exmp}
\par  The following example illustrates that an optimal encoding matrix for $\mathcal{I}_{E}$ can be obtained even if sub-optimal encoding matrices are used for some of the constituent SGICPs, as long as the conditions given in Theorem \ref{th1} are satisfied.\\ 
\begin{exmp}
Consider the jointly extended problem given by ${\bf{F}}^{E}_{x}$ as shown below, with the sub-problems being the same as given in Example \ref{ex1}. We consider the binary field $\mathbb{F}_2$.
\[ {\bf{F}}^{E}_{x} = \left(
\begin{array}{ccccc|ccc}
1 & x & x & 0 & 0 & x & 0 & 0 \\
0 & 1 & x & x & 0 & 0 & x & 0 \\
0 & 0 & 1 & x & x & 0 & 0 & x\\
x & 0 & 0 & 1 & x & x & 0 & x\\
x & x & 0 & 0 & 1 & x & x & x\\
\hline
x & 0 & x & 0 & 0 & 1 & 0 & x\\
0 & x & 0 & x & 0 & x & 1 & 0 \\
0 & 0 & x & x & x & 0 & x & 1 \\
\end{array} \right).
\]
Based on the first three rows of ${\bf{B}}_{x_0}^{(12)}$ we take ${\bf{F}}^{(2)}$ as shown below. Taking
${\bf{F}}^{(1)}$ and ${\bf{P}}^{(1)}$ to be the same as in Example \ref{ex1}, all the conditions of Theorem \ref{th1} are satisfied. Note that ${\bf{P}}^{(2)}$ does not exist. 
\[ {\bf{F}}^{(2)} =
\begin{pmatrix}
1 & 0 & 0 \\
0 & 1 & 0 \\
0 & 0 & 1 \\
\end{pmatrix}.
\]
\par Note that  $\text{mrk}_{q}({\bf{F}}_x^{(2)}) = 2$, but $r_2 = 3$. As $\text{mrk}_{q}({\bf{F}}_x^{(1)}) = 3 \geq r_2$, we can obtain an optimal scalar linear index code for the jointly extended problem according to Theorem \ref{th1}.  The optimal encoding matrix ${\bf{G}}_{E}$ as given by Theorem \ref{th1} is shown below. 
\[ {\bf{G}}_{E} = \left(
\begin{array}{ccccc|ccc}
1 & 0 & 1 & 0 & 0 & 1 & 0 & 0\\
0 & 1 & 0 & 1 & 0 & 0 & 1 & 0\\
0 & 0 & 1 & 1 & 1 & 0 & 0 & 1
\end{array} \right).
\]
\label{exmp1}
\end{exmp}
%%%%%%%%%%%%%%%%%%%%%%%%%%%%%%%%%%%%%%%%%
\section{Optimal Scalar linear Codes for TGICPs belonging to Cases II-B, II-C, II-D, and II-E}
%%%%%%%%%%%%%%%%%%%%%%%%%%%%
In this section, we apply the results presented in the previous section to obtain upper bounds on the optimal scalar linear code lengths for some classes of the TGICP using code constructions. We then provide some necessary conditions for the optimality of the constructed codes. 

The following lemma provides a lower bound on the optimal scalar linear code length of any TGICP.

\begin{lem}
	\label{noP12}
	For any TGICP $\mathcal{I}$ with $\mathcal{P}_{\{1,2\}}=\phi$,
	\begin{gather*} l^*_{q}(\mathcal{I})=\text{mrk}_{q}({\bf{F}}_x^{(1)})+\text{mrk}_{q}({\bf{F}}_x^{(2)}).
	\end{gather*}
\end{lem}
\begin{proof}
	By transmitting optimal codes for the SGICPs $\mathcal{I}_{1}$ and $\mathcal{I}_{2}$, we have  $l^*_{q}(\mathcal{I}) \leq \text{mrk}_{q}(\mathcal{I}_{1})+\text{mrk}_{q}(\mathcal{I}_{2})$. If $\mathcal{S}_{j}$ transmits a codeword ${\bf{c}}^{(j)}$, $j \in \{1,2\}$, it does not contain messages from ${\bf{x}}_{\mathcal{P}_{k}}$ where $k = \{1,2\} \setminus j$. Thus, receivers in $\mathcal{I}_{j}$ can not take advantage of the side information present in ${\bf{x}}_{\mathcal{P}_{k}}$, $k = \{1,2\} \setminus j$, $j \in \{1,2\}$. Hence, an optimal scalar linear index code for the TGICP $\mathcal{I}$ must consist of optimal codes for the SGICPs $\mathcal{I}_{{1}}$ and $\mathcal{I}_{{2}}$. Thus, $l^*_{q}(\mathcal{I}) \geq \text{mrk}_{q}(\mathcal{I}_{1})+\text{mrk}_{q}(\mathcal{I}_{2})$.
\end{proof}

The following observation is used to obtain the optimal scalar linear codes for a TGICP using those of a related TGICP.

\begin{obs}
	Consider the TGICP $\mathcal{I}$ with the submatrices of the related fitting matrix ${\bf{F}}_x$ being $\{{\bf{F}}_x^{(\mathcal{A})}\}$, $\{{\bf{A}}^{(\mathcal{A}_1,\mathcal{A}_2)}_{x_0}\}$, for non-empty $\mathcal{A}$,  $\mathcal{A}_1$, $\mathcal{A}_2 \subseteq \{1,2\}$, $\mathcal{A}_1 \neq \mathcal{A}_2$. Consider another TGICP $\mathcal{\tilde{I}}$ with the submatrices of the related fitting matrix ${\bf{\tilde{F}}}_x$ given by  $\{{\bf{\tilde{F}}}_x^{(\mathcal{A})}\}$, $\{{\bf{\tilde{A}}}^{(\mathcal{A}_1,\mathcal{A}_2)}_{x_0}\}$, for non-empty $\mathcal{A}$,  $\mathcal{A}_1$, $\mathcal{A}_2 \subseteq \{1,2\}$, $\mathcal{A}_1 \neq \mathcal{A}_2$.  Let  ${\bf{F}}_x^{(1)}={\bf{\tilde{F}}}_x^{(2)}$, ${\bf{F}}_x^{(2)}={\bf{\tilde{F}}}_x^{(1)}$, and ${\bf{F}}_x^{\{1,2\}}={\bf{\tilde{F}}}_x^{\{1,2\}}$. Also, ${\bf{A}}^{(\mathcal{A}_1,\{1,2\})}_{x_0}={\bf{\tilde{A}}}^{(\{1,2\}\setminus \mathcal{A}_1,\{1,2\})}_{x_0}$, and ${\bf{A}}^{(\{1,2\},\mathcal{A}_1)}_{x_0}={\bf{\tilde{A}}}^{(\{1,2\},\{1,2\}\setminus \mathcal{A}_1)}_{x_0}$, for all $\mathcal{A}_1 \subset \{1,2\}$, such that $|\mathcal{A}_1|=1$. If the two-sender index code for $\mathcal{I}$ consists of ${\bf{c}}^{(i)}$ transmitted by $\mathcal{S}_i$, for all $i \in \{1,2\}$, then the code ${\bf{c}}^{(\{1,2\}\setminus i)}$ transmitted by $\mathcal{S}_{i}$ for all $i \in \{1,2\}$, is a two-sender index code for $\mathcal{\tilde{I}}$. This is because the TGICP $\mathcal{I}$ is obtained from the TGICP $\mathcal{\tilde{I}}$ by exchanging the message sets available with the senders as seen from the fitting matrices ${\bf{F}}_x$ and ${\bf{\tilde{F}}}_x$. Moreover, the set of receivers in $\mathcal{I}$ having demands available with a particular sender, have the same demands available with the other sender in $\mathcal{\tilde{I}}$.
	\label{obs2} 
\end{obs}	

\subsection{Case II-B}
In this section, we provide optimal scalar linear codes for some classes of the TGICP belonging to Case II-B.

\begin{thm}
	Consider any TGICP  $\mathcal{I}$ belonging to Case II-B, with the interactions between $\mathcal{I}_1$ and $\mathcal{I}_{\{1,2\}}$, and those between $\mathcal{I}_2$ and $\mathcal{I}_{\{1,2\}}$, being fully-participated interactions. Then we have, $l_{q}^*(\mathcal{I}) =   \max\{\text{mrk}_{q}({\bf{F}}^{(1)}_x)+\text{mrk}_{q}({\bf{F}}_x^{(2)}),\text{mrk}_{q}({\bf{F}}^{(\{1,2\})}_x)\}$.
	\label{thcase2b1}
\end{thm}
\begin{proof}
	The proof is based on the code construction given in the proof of Theorem $7$, in \cite{CTLO}. The lower bound $l_{q}^*(\mathcal{I}) \geq   \max\{\text{mrk}_{q}({\bf{F}}^{(1)}_x)+\text{mrk}_{q}({\bf{F}}_x^{(2)}),\text{mrk}_{q}({\bf{F}}^{(\{1,2\})}_x)\}$ is obtained by using Lemma \ref{noP12} and observing that $\mathcal{I}_{\{1,2\}}$ is a sub-problem. The matching upper bound is based on a code construction. There are four sub-cases: $(i) ~\text{mrk}_{q}({\bf{F}}^{(\{1,2\})}_x) \geq \text{mrk}_{q}({\bf{F}}^{(1)}_x) + \text{mrk}_{q}({\bf{F}}^{(2)}_x)$, $(ii)~\text{mrk}_{q}({\bf{F}}^{(\{1,2\})}_x)  \geq \max\{\text{mrk}_{q}({\bf{F}}^{(1)}_x),\text{mrk}_{q}({\bf{F}}^{(2)}_x)\}$, $(iii)$ $\text{mrk}_{q}({\bf{F}}^{(\{1,2\})}_x) \leq \text{mrk}_{q}({\bf{F}}^{(1)}_x)$, and $(iv)~\text{mrk}_{q}({\bf{F}}^{(\{1,2\})}_x) \leq \text{mrk}_{q}({\bf{F}}^{(2)}_x)$. Let ${\bf{c}}^{(\mathcal{A})}$ be an optimal scalar linear code for $\mathcal{I}_{\mathcal{A}}$, for any non-empty $\mathcal{A} \subseteq \{1,2\}$. 	Code construction for sub-cases $(i)$ and $(ii)$: ${\bf{c}}^{(1)} + {\bf{c}}^{(\{1,2\})}[1:\text{mrk}_{q}({\bf{F}}^{(1)}_x)]$ transmitted by $\mathcal{S}_1$, and ${\bf{c}}^{(2)} + {\bf{c}}^{(\{1,2\})}[\text{mrk}_{q}({\bf{F}}^{(1)}_x)+1:\text{mrk}_{q}({\bf{F}}^{(\{1,2\})}_x)]$ transmitted by $\mathcal{S}_2$. Code construction for sub-case $(iii)$: ${\bf{c}}^{(1)} + {\bf{c}}^{(\{1,2\})}$ transmitted by $\mathcal{S}_1$, and ${\bf{c}}^{(2)}$ transmitted by $\mathcal{S}_2$.  Code construction for sub-case $(iv)$: ${\bf{c}}_2 + {\bf{c}}^{(\{1,2\})}$ transmitted by $\mathcal{S}_2$, and ${\bf{c}}^{(1)}$ transmitted by $\mathcal{S}_1$.
\end{proof}

\begin{thm}
	Consider any TGICP  $\mathcal{I}$ belonging to Case II-B. Let the SGICP $\mathcal{\tilde{I}}_i$, $i \in \{1,2\}$, given by the partitioned fitting matrix ${\bf{\tilde{F}}}_x^{(i)}$ in (\ref{eqcase2b2}) be a joint extension of $\mathcal{I}_i$ and $\mathcal{I}_{\{1,2\}}$ given by Theorem \ref{th1} with the matrix ${\bf{F}}^{(\{1,2\})} \approx {\bf{F}}^{(\{1,2\})}_x$ being the same for  $\mathcal{\tilde{I}}_1$ and $\mathcal{\tilde{I}}_2$. 
	\begin{equation}
	{\bf{\tilde{F}}}_x^{(i)} =
	\left(
	\begin{array}{c|c} 
	{\bf{F}}_x^{(i)}  &  {\bf{A}}^{(i,\{1,2\})}_{x_0}\\
	\hline
	{\bf{A}}^{(\{1,2\},i)}_{x_0} &  {\bf{F}}_x^{(\{1,2\})}
	\end{array}
	\right).
	\label{eqcase2b2}
	\end{equation}
	\label{corcase2b2}
	$(i)$ If $\text{mrk}_{q}({\bf{F}}_x^{(\{1,2\})}) > \max\{\text{mrk}_{q}({\bf{F}}_x^{(1)}),\text{mrk}_{q}({\bf{F}}_x^{(2)})\}$, we have  $l_{q}^*(\mathcal{I}) \leq \text{mrk}_{q}({\bf{F}}_x^{(\{1,2\})})+\min\{\text{mrk}_{q}({\bf{F}}_x^{(1)}),\text{mrk}_{q}({\bf{F}}_x^{(2)})\}$.\\
	$(ii)$ Otherwise,  $l_{q}^*(\mathcal{I}) = \text{mrk}_{q}({\bf{F}}_x^{(1)})+\text{mrk}_{q}({\bf{F}}_x^{(2)})$.
	\label{thcase2b2}
\end{thm}
\begin{proof}
	Let $j = \underset{i \in \{1,2\}}{\text{argmax}} \{\text{mrk}_{q}({\bf{F}}_x^{(i)})\}$. We first prove $(i)$. As the sub-problem with the fitting matrix ${\bf{\tilde{F}}}^{(j)}_x$ can be considered as an SGICP, $\mathcal{S}_{j}$ can transmit an optimal scalar linear code of length $\max\{\text{mrk}_{q}({\bf{F}}^{(j)}_x),\text{mrk}_{q}({\bf{F}}^{(\{1,2\})}_x)\}=\text{mrk}_{q}({\bf{F}}^{(\{1,2\})}_x)$. This satisfies the receivers in $\mathcal{I}_{j}$ and $\mathcal{I}_{\{1,2\}}$.
	To satisfy the receivers in $\mathcal{I}_{j^c}$, where $j^c = \{1,2\} \setminus j$, $\mathcal{S}_{j^c}$ can transmit an optimal scalar linear code of length $\text{mrk}_{q}({\bf{F}}^{(j^c)}_x)=\min\{\text{mrk}_{q}({\bf{F}}_x^{(1)}),\text{mrk}_{q}({\bf{F}}_x^{(2)})\}$. Hence, we have the result of statement $(i)$ as an upper bound.
	
	We now prove $(ii)$. As given in the proof of $(i)$, $\mathcal{S}_{j}$ can transmit an optimal scalar linear code of length $\max\{\text{mrk}_{q}({\bf{F}}^{(j)}_x),\text{mrk}_{q}({\bf{F}}^{(\{1,2\})}_x)\}=\text{mrk}_{q}({\bf{F}}^{(j)}_x)$. This satisfies the receivers in $\mathcal{I}_{j}$ and $\mathcal{I}_{\{1,2\}}$.
	To satisfy the receivers in $\mathcal{I}_{j^c}$, where $j^c = \{1,2\} \setminus j$, $\mathcal{S}_{j^c}$ can transmit an optimal scalar linear code of length $\text{mrk}_{q}({\bf{F}}_x^{(j^c)})$. Hence, this upper bound matches the lower bound given by Lemma \ref{noP12}.
\end{proof}

We now illustrate Theorem \ref{thcase2b2} with an example.

\begin{exmp}
	Consider the partitioned fitting matrix ${\bf{F}}_x$ related to the TGICP $\mathcal{I}$ as given below. We choose $\mathbb{F}_2$. Let ${\bf{x}}_{\mathcal{M}_1}=\{{\bf{x}}_1,{\bf{x}}_2,{\bf{x}}_3,{\bf{x}}_4,{\bf{x}}_8,{\bf{x}}_9\}$, ${\bf{x}}_{\mathcal{M}_2}=\{{\bf{x}}_5,{\bf{x}}_6,{\bf{x}}_7,{\bf{x}}_8,{\bf{x}}_9\}$.
	\[ {\bf{F}}_{x} = \left(
	\begin{array}{cccc|ccc|cc}
	1 & x & x & 0 & 0 & x & x & x & x\\
	0 & 1 & x & x & 0 & x & x & 0 & 0\\
	x & 0 & 1 & x & x & 0 & 0 & x & x\\
	x & x & 0 & 1 & x & x & x & x & x\\
	\hline
	x & x & 0 & 0 & 1 & x & 0 & x & x\\
	x & 0 & x & 0 & 0 & 1 & x & 0 & 0\\
	0 & x & x & x & x & 0 & 1 & x & x\\
	\hline
	x & 0 & x & 0 & x & x & 0 & 1 & x\\
	x & x & x & x & x & x & x & x & 1\\
	\end{array} \right).
	\]
	The sub-problems $\mathcal{\tilde{I}}_1$ and  $\mathcal{\tilde{I}}_2$ given by the partitioned fitting matrices ${\bf{\tilde{F}}}_x^{(1)}$ and ${\bf{\tilde{F}}}_x^{(2)}$ are joint extensions of $\mathcal{I}_1$ and $\mathcal{I}_{\{1,2\}}$, and  $\mathcal{I}_2$ and $\mathcal{I}_{\{1,2\}}$ respectively as given by Theorem \ref{th1}.
	\[ {\bf{\tilde{F}}}_{x}^{(1)} = \left(
		\begin{array}{cccc|cc}
		1 & x & x & 0 & x & x\\
		0 & 1 & x & x & 0 & 0\\
		x & 0 & 1 & x & x & x\\
		x & x & 0 & 1 & x & x\\
		\hline
		x & 0 & x & 0 & 1 & x\\
		x & x & x & x & x & 1\\
		\end{array} \right).
	\]
		\[ {\bf{\tilde{F}}}_{x}^{(2)} = \left(
		\begin{array}{ccc|cc}
		1 & x & 0 & x & x\\
		0 & 1 & x & 0 & 0\\
		x & 0 & 1 & x & x\\
		\hline
		x & x & 0 & 1 & x\\
		x & x & x & x & 1\\
		\end{array} \right).
		\]	
    The completions of the fitting matrices of the sub-problems of the TGICP $\mathcal{I}$, in accordance with Theorem \ref{thcase2b2} are given below.
	\[	
	{\bf{F}}^{(1)} = \left(
	\begin{array}{cccc}
	1 & 0 & 1 & 0\\
	0 & 1 & 0 & 1\\
	\hline
	1 & 0 & 1 & 0\\
	0 & 1 & 0 & 1\\
	\end{array} \right),
	{\bf{F}}^{(\{1,2\})} = \left(
	\begin{array}{cc}
	1 & 1\\
	\hline
	1 & 1\\
	\end{array} \right).
	\]		
	\[	
	{\bf{F}}^{(2)} = \left(
	\begin{array}{ccc}
	1 & 1 & 0 \\
	0 & 1 & 1 \\
	\hline
	1 & 0 & 1\\
	\end{array} \right).
	\]	
	Note that $\text{mrk}_q({\bf{F}}^{(1)}_x)=2$ (from \cite{MBR1}), $\text{mrk}_q({\bf{F}}^{(2)}_x)=2$, $\text{mrk}_q({\bf{F}}^{(\{1,2\})}_x)=1$. Hence, $l^*_q(\mathcal{I})=4$. An optimal scalar linear code given according to Theorem \ref{thcase2b2} is : $\mathcal{S}_1$ sends $(~{\bf{x}}_1+{\bf{x}}_3+{\bf{x}}_8+{\bf{x}}_9, ~{\bf{x}}_2+{\bf{x}}_4~)$, and $\mathcal{S}_2$ sends $(~{\bf{x}}_5+{\bf{x}}_6, ~{\bf{x}}_6+{\bf{x}}_7~)$.	
\end{exmp}

The following corollary  establishes the optimal scalar linear code length for another class of the TGICP belonging to Case II-B.  The proof follows on similar lines as that of result $(ii)$ in  Theorem \ref{thcase2b2} in conjunction with Observation \ref{obs1}.

\begin{cor}
	Consider any TGICP  $\mathcal{I}$ belonging to Case II-B. Let any one of the SGICPs given by the partitioned fitting matrices ${\bf{\tilde{F}}}_x^{(i)}$, $i \in \{1,2\}$, given in (\ref{eqcase2b2}) be a rank-invariant extension that can be derived as a special case of Theorem \ref{th1} as given in Observation \ref{obs1}. Then we have, $l_{q}^*(\mathcal{I}) =   \text{mrk}_{q}({\bf{F}}_x^{(1)})+\text{mrk}_{q}({\bf{F}}^{(2)}_x)$.
	\label{corcase2b3}
\end{cor}

\subsection{Cases II-C and II-D}
In this section, we provide optimal scalar linear codes for some classes of the TGICP belonging to Cases II-C and II-D.

\begin{thm}
Consider any TGICP  $\mathcal{I}$ belonging to Case II-C. Let the SGICP given by the partitioned fitting matrix ${\bf{\tilde{F}}}_x$ in (\ref{eqcase2c}) be a joint extension of $\mathcal{I}_1$ and $\mathcal{I}_{\{1,2\}}$ given by Theorem \ref{th1}, or rank-invariant extensions that can be derived as special cases of Theorem \ref{th1} as given in Corollary \ref{corth1} and Observation \ref{obs1}.  Then we have, $l_{q}^*(\mathcal{I}) =   \text{mrk}_{q}({\bf{F}}_x^{(2)})+\max\{\text{mrk}_{q}({\bf{F}}^{(1)}_x),\text{mrk}_{q}({\bf{F}}^{(\{1,2\})}_x)\}$.
\begin{equation}
{\bf{\tilde{F}}}_x =
\left(
\begin{array}{c|c} 
{\bf{F}}_x^{(1)}  &  {\bf{A}}^{(1,\{1,2\})}_{x_0}\\
\hline
{\bf{A}}^{(\{1,2\},1)}_{x_0} &  {\bf{F}}_x^{(\{1,2\})}
\end{array}
\right).
\label{eqcase2c}
\end{equation}
\label{thcase2c}
\end{thm}
\begin{proof}
We first consider the TGICP $\mathcal{I}_f$  with  ${\bf{A}}^{(1,\{1,2\})}_{x_0}={\bf{A}}^{(\{1,2\},1)}_{x_0}={\bf{X}}$, and all other submatrices of the related fitting matrix being the same as that of the TGICP $\mathcal{I}$. The subscript $f$ in $\mathcal{I}_f$ denotes fully-participated interactions between $\mathcal{I}_1$ and $\mathcal{I}_{\{1,2\}}$. Note that $l_q^*(\mathcal{I}_f) \leq l_q^*(\mathcal{I})$, as each receiver in $\mathcal{I}$ has side information which is a subset of that of the corresponding receiver in $\mathcal{I}_f$. We now provide an upper bound for  $l_q^*(\mathcal{I}_f)$, and then a matching lower bound. 

Let $\mathcal{S}_2$ transmit an optimal scalar linear code for $\mathcal{I}_2$  with length $\text{mrk}_{q}({\bf{F}}_x^{(2)})$. Note that all the receivers in $\mathcal{I}_2$ are satisfied. Using Corollary \ref{corth1}, the optimal scalar linear code length for the SGICP with the fitting matrix ${\bf{\tilde{F}}}_x$ is $\max\{\text{mrk}_{q}({\bf{F}}^{(1)}_x),\text{mrk}_{q}({\bf{F}}^{(\{1,2\})}_x)\}$. All receivers in $\mathcal{I}_1$ and $\mathcal{I}_{\{1,2\}}$ are satisfied. Hence, $l_{q}^*(\mathcal{I}_f) \leq   \text{mrk}_{q}({\bf{F}}_x^{(2)})+\max\{\text{mrk}_{q}({\bf{F}}^{(1)}_x),\text{mrk}_{q}({\bf{F}}^{(\{1,2\})}_x)\}$. Using Lemma \ref{noP12}, note that
\begin{equation}
l_{q}^*(\mathcal{I}_f) \geq \text{mrk}_{q}({\bf{F}}_x^{(1)})+\text{mrk}_{q}({\bf{F}}_x^{(2)}).
\label{eqsp12}
\end{equation}
Consider the two-sender sub-problem $\mathcal{I}'$ of $\mathcal{I}_f$  given by the fitting matrix  ${\bf{F}}'_x$ in (\ref{eqfitpartsp}). As the TGICP $\mathcal{I}_f$ belongs  to Case II-C, at least one of  ${\bf{A}}_{x_0}^{(2,\{1,2\})}$ and ${\bf{A}}^{(\{1,2\},2)}_{x_0}$ is ${\bf{0}}$.  
\begin{equation}
{\bf{F}}'_x =
\left(
\begin{array}{c|c} 
{\bf{F}}_x^{(2)} & {\bf{A}}_{x_0}^{(2,\{1,2\})} \\
\hline
{\bf{A}}^{(\{1,2\},2)}_{x_0} & {\bf{F}}_x^{(\{1,2\})}
\end{array}
\right).
\label{eqfitpartsp}
\end{equation}
Using Lemma \ref{lmSCC} and noting that the TGICP $\mathcal{I}'$ can be considered as an SGICP with message set ${\bf{x}}_{\mathcal{M}_2}$, we have
\begin{equation}
l_{q}^*(\mathcal{I}_f) \geq l_{q}^*(\mathcal{I}')= \text{mrk}_{q}({\bf{F}}_x^{(2)})+\text{mrk}_{q}({\bf{F}}_x^{(\{1,2\})}).
\label{eqsp121}
\end{equation}
Combining (\ref{eqsp12}) and (\ref{eqsp121}), we obtain the matching lower bound $l_{q}^*(\mathcal{I}_f) \geq   \text{mrk}_{q}({\bf{F}}_x^{(2)})+\max\{\text{mrk}_{q}({\bf{F}}^{(1)}_x),\text{mrk}_{q}({\bf{F}}^{(\{1,2\})}_x)\}$. This serves as a lower bound for $l_{q}^*(\mathcal{I})$. The matching upper bound is obtained as in the case of the TGICP $\mathcal{I}_f$ using Theorem \ref{th1} (this paper) or Theorems $1$ and $2$ in \cite{PK}.
\end{proof}

We now illustrate Theorem \ref{thcase2c} with an example.

\begin{exmp}
	Consider the partitioned fitting matrix ${\bf{F}}_x$ related to the TGICP $\mathcal{I}$ as given below. We choose $\mathbb{F}_2$. Let ${\bf{x}}_{\mathcal{M}_1}=\{{\bf{x}}_1,{\bf{x}}_2,{\bf{x}}_3,{\bf{x}}_4,{\bf{x}}_5,{\bf{x}}_8,{\bf{x}}_9\}$, ${\bf{x}}_{\mathcal{M}_2}=\{{\bf{x}}_6,{\bf{x}}_7,{\bf{x}}_8,{\bf{x}}_9\}$.
	\[ {\bf{F}}_{x} = \left(
	\begin{array}{ccccc|cc|cc}
	1 & x & x & 0 & 0 & x & 0 & x & x\\
	0 & 1 & x & x & 0 & 0 & x & 0 & 0\\
	0 & 0 & 1 & x & x & 0 & 0 & x & x\\
	x & 0 & 0 & 1 & x & 0 & x & x & x\\
	x & x & 0 & 0 & 1 & x & 0 & x & x\\
	\hline
	x & 0 & 0 & 0 & 0 & 1 & x & 0 & 0\\
	0 & x & 0 & x & x & 0 & 1 & x & x\\
	\hline
	x & x & x & 0 & 0 & 0 & 0 & 1 & x\\
	x & x & x & 0 & x & 0 & 0 & x & 1\\
	\end{array} \right).
	\]
	With the completions of the fitting matrices of the sub-problems of the TGICP $\mathcal{I}$ given below, we see that the SGICP with fitting matrix ${\bf{\tilde{F}}}_x$ (given in Theorem \ref{thcase2c}) is a joint extension given in Theorem \ref{th1}.
	\[	
	{\bf{F}}^{(1)} = \left(
	\begin{array}{ccccc}
	1 & 1 & 1 & 0 & 0\\
	0 & 1 & 0 & 1 & 0\\
	0 & 0 & 1 & 0 & 1\\	
	\hline
	1 & 0 & 0 & 1 & 1\\
	1 & 1 & 0 & 0 & 1\\
	\end{array} \right),
	{\bf{F}}^{(\{1,2\})} = \left(
	\begin{array}{cc}
	1 & 1\\
	\hline
	1 & 1\\
	\end{array} \right).
	\]		
	\[	
	{\bf{F}}^{(2)} = \left(
	\begin{array}{cc}
	1 & 0 \\
	0 & 1 \\
	\end{array} \right).
	\]	
	Note that $\text{mrk}_q({\bf{F}}^{(1)}_x)=3$ (from \cite{MBR1}), $\text{mrk}_q({\bf{F}}^{(2)}_x)=2$, $\text{mrk}_q({\bf{F}}^{(\{1,2\})}_x)=1$. Hence, $l^*_q(\mathcal{I})=5$. An optimal scalar linear code given according to Theorem \ref{thcase2c} is : $\mathcal{S}_1$ sends $(~{\bf{x}}_1+{\bf{x}}_2+{\bf{x}}_3+{\bf{x}}_8+{\bf{x}}_9, ~{\bf{x}}_2+{\bf{x}}_4,~{\bf{x}}_3+{\bf{x}}_5~)$, and $\mathcal{S}_2$ sends $({\bf{x}}_6, ~{\bf{x}}_7)$.	
\end{exmp}

The following corollary  establishes the optimal scalar linear code length for some classes of the TGICP belonging to Case II-D.  The result follows directly from Theorem \ref{thcase2c} in conjunction with Observation \ref{obs2}.

\begin{cor}
	Consider any TGICP  $\mathcal{I}$ belonging to Case II-D. Let the SGICP given by the partitioned fitting matrix ${\bf{\tilde{F}}}_x$ in (\ref{eqcase2d}) be a joint extension of $\mathcal{I}_2$ and $\mathcal{I}_{\{1,2\}}$ given by Theorem \ref{th1}, or rank-invariant extensions that can be derived as special cases of Theorem \ref{th1} as given in Corollary \ref{corth1} and Observation \ref{obs1}.  Then we have, $l_{q}^*(\mathcal{I}) =   \text{mrk}_{q}({\bf{F}}_x^{(1)})+\max\{\text{mrk}_{q}({\bf{F}}^{(2)}_x),\text{mrk}_{q}({\bf{F}}^{(\{1,2\})}_x)\}$.
	\begin{equation}
	{\bf{\tilde{F}}}_x =
	\left(
	\begin{array}{c|c} 
	{\bf{F}}_x^{(2)}  &  {\bf{A}}^{(2,\{1,2\})}_{x_0}\\
	\hline
	{\bf{A}}^{(\{1,2\},2)}_{x_0} &  {\bf{F}}_x^{(\{1,2\})}
	\end{array}
	\right).
	\label{eqcase2d}
	\end{equation}
	\label{corcase2d}
\end{cor}

\subsection{Case II-E}

In this section, we provide scalar linear codes for some classes of the TGICP belonging to Case II-E, and give some necessary conditions for their optimality.

\begin{thm}
	Consider any TGICP  $\mathcal{I}$ belonging to Case II-E, with all  interactions between $\mathcal{I}_1$, $\mathcal{I}_2$, and $\mathcal{I}_{\{1,2\}}$  being fully-participated. Then we have, $l_{q}^*(\mathcal{I})=$ $\max\{\text{mrk}_{q}({\bf{F}}^{(1)}_x)+\text{mrk}_{q}({\bf{F}}_x^{(2)}),~\text{mrk}_{q}({\bf{F}}^{(1)}_x)+\text{mrk}_{q}({\bf{F}}^{(\{1,2\})}_x),~\text{mrk}_{q}({\bf{F}}^{(2)}_x)+\text{mrk}_{q}({\bf{F}}^{(\{1,2\})}_x)\}$.
	\label{thcase2b1}
\end{thm}
\begin{proof}
The proof follows on similar lines as that of Theorem $6$ in \cite{CVBSR1}. We first provide an upper bound based on code construction and then provide a lower bound.

Let ${\bf{c}}^{(\mathcal{A})}$ be an optimal scalar linear code for $\mathcal{I}_{\mathcal{A}}$, for any non-empty $\mathcal{A} \subseteq \{1,2\}$.
Consider the case when $\text{mrk}_{q}({\bf{F}}^{(\{1,2\})}_x) \leq \min\{\text{mrk}_{q}({\bf{F}}^{(1)}_x),\text{mrk}_{q}({\bf{F}}_x^{(2)})\}$. If $\mathcal{S}_1$ transmits ${\bf{c}}^{(1)}+{\bf{c}}^{(\{1,2\})}$, and $\mathcal{S}_2$ transmits ${\bf{c}}^{(2)}+{\bf{c}}^{(\{1,2\})}$, we obtain the optimal code length given in the theorem. Consider the case when $\text{mrk}_{q}({\bf{F}}^{(j)}_x) \leq \min\{\text{mrk}_{q}({\bf{F}}^{(j^c)}_x),\text{mrk}_{q}({\bf{F}}_x^{(\{1,2\})})\}$, for any $j \in \{1,2\}$, $j^c=\{1,2\} \setminus j$. If $\mathcal{S}_j$ transmits ${\bf{c}}^{(j)}+{\bf{c}}^{(\{1,2\})}[1:\text{mrk}_{q}({\bf{F}}^{(j)}_x)]$, and $\mathcal{S}_{j^c}$ transmits ${\bf{c}}^{(j^c)}+{\bf{c}}^{(\{1,2\})}[1:\text{mrk}_{q}({\bf{F}}^{(j)}_x)]$,  ${\bf{c}}^{(\{1,2\})}[\text{mrk}_{q}({\bf{F}}^{(j)}_x)+1:\text{mrk}_{q}({\bf{F}}^{(\{1,2\})}_x)]$, we obtain the optimal code length given in the theorem. Decodability at receivers is on similar lines as that given in Theorem $6$ in \cite{CVBSR1}. Now we provide a matching lower bound.

Consider any TGICP $\mathcal{I}$ with all existing interactions being fully-participated and belonging to Case II-E. We obtain $(i)$ a TGICP $\mathcal{I}'$ with all existing interactions being fully-participated, and whose interaction digraph is either $\mathcal{H}_{44}$, or $\mathcal{H}_{45}$, and $(ii)$ a TGICP $\mathcal{I}''$ with all existing interactions being fully-participated, and whose interaction digraph is either $\mathcal{H}_{56}$, or $\mathcal{H}_{57}$, by adding appropriate fully-participated interactions among the sub-problems of $\mathcal{I}$. Hence, we have $l^*_q(\mathcal{I}) \geq l^*_q(\mathcal{I}'),l^*_q(\mathcal{I}'')$. Thus combining the results of Lemma \ref{noP12}, Theorem \ref{thcase2c}, and Corollary \ref{corcase2d}, we have the lower bound equal to $l^*_q(\mathcal{I})$ as stated in this theorem. 
\end{proof}

\begin{thm}
	Consider any TGICP  $\mathcal{I}$ belonging to Case II-E with the interaction digraph being  $\mathcal{H}_k$, where $k \in \{58,59,60\}$. Let $j \in \{1,2\}$, such that the interaction $\mathcal{I}_{\{1,2\}} \rightarrow \mathcal{I}_{j}$ exists, and   $\text{mrk}_{q}({\bf{F}}^{(\{1,2\})}_x) \geq \text{mrk}_{q}({\bf{F}}^{(j)}_x)$. Let the SGICP $\mathcal{\tilde{I}}$ given by the partitioned fitting matrix ${\bf{\tilde{F}}}_x$ in (\ref{eqcase2d2}) be a joint extension of $\mathcal{I}_1$ and $\mathcal{I}_{2}$ given by Theorem \ref{th1}, with optimal completions ${\bf{F}}^{(1)}$ and ${\bf{F}}^{(2)}$. Let $r_{\mathcal{A}}=\text{mrk}_q({\bf{F}}_x^{(\mathcal{A})})$, for non-empty $\mathcal{A} \subseteq \{1,2\}$. 
	\begin{equation}
	{\bf{\tilde{F}}}_x =
	\left(
	\begin{array}{c|c} 
	{\bf{F}}_x^{(1)}  &  {\bf{A}}^{(1,2)}_{x_0}\\
	\hline
	{\bf{A}}^{(2,1)}_{x_0} &  {\bf{F}}_x^{(2)}
	\end{array}
	\right).
	\label{eqcase2d2}
	\end{equation}
	 Let ${\bf{F}}^{(\{1,2\})} \approx {\bf{F}}^{(\{1,2\})}_x$ be such that the first $r_{\{1,2\}}$ rows of ${\bf{F}}^{(\{1,2\})}$ span $\langle {\bf{F}}^{(\{1,2\})} \rangle$.  Let ${\bf{A}}^{(\{1,2\},j)}_{x_0}$ be any matrix such that ${\bf{A}}^{(\{1,2\},j)} \approx {\bf{A}}^{(\{1,2\},j)}_{x_0}$ where  
	${\bf{A}}^{(\{1,2\},j)}=\left( \begin{array}{c}  {\bf{\hat{F}}}^{(j)} \\
	\hline
	{\bf{P}}^{(\{1,2\})}{\bf{\hat{F}}}^{(j)}
	\end{array} \right)$, where 
	${\bf{\hat{F}}}^{(j)} =
	\left( \begin{array}{c} {\bf{F}}^{(j)}_{[r_j]}\\
	\hline
	{\bf{0}}_{(r_{\{1,2\}}-r_j) \times m_j}\\
	\end{array}
	\right)$, and ${\bf{P}}^{(\{1,2\})}$ be the $(n_{\{1,2\}}-r_{\{1,2\}}) \times r_{\{1,2\}}$ matrix such that the last $(n_{\{1,2\}}-r_{\{1,2\}})$ rows of ${\bf{F}}^{(\{1,2\})}$ is given by  ${\bf{P}}^{(\{1,2\})}{\bf{F}}^{(\{1,2\})}_{[r_{\{1,2\}}]}$. Then we have, $(i)~$ $l_{q}^*(\mathcal{I}) \leq \text{mrk}_{q}({\bf{F}}_{x}^{(\{1,2\})})+\max\{\text{mrk}_{q}({\bf{F}}_{x}^{(j^c)}),\text{mrk}_{q}({\bf{F}}_{x}^{(\{1,2\})})\}$.
	If $\max\{\text{mrk}_{q}({\bf{F}}_{x}^{(j^c)}),\text{mrk}_{q}({\bf{F}}_{x}^{(\{1,2\})})\}=\text{mrk}_{q}({\bf{F}}_{x}^{(j^c)})$, then $(ii)$ $l_{q}^*(\mathcal{I})=\text{mrk}_{q}({\bf{F}}_{x}^{(j^c)})+\text{mrk}_{q}({\bf{F}}_{x}^{(\{1,2\})})$.
	\label{thcase2e2}
\end{thm}
\begin{proof}
	Consider the sub-problem with fitting matrix ${\bf{F}}'_x=( {\bf{F}}^{(\{1,2\})}_x | {\bf{A}}^{(\{1,2\},j)}_{x_0} )$. Consider the optimal scalar linear code ${\bf{c}}^{(\mathcal{A})}={\bf{F}}^{(\mathcal{A})}_{[r_{\mathcal{A}}]} {\bf{x}}_{\mathcal{P}_{\mathcal{A}}}$, of sub-problem  $\mathcal{I}_{\mathcal{A}}$, for non-empty $\mathcal{A} \subseteq \{1,2\}$. Let ${\bf{D}}=(({\bf{I}}_{r_{\{1,2\}} \times r_{\{1,2\}}})^T|({\bf{P}}^{(\{1,2\} )})^T)^T$. It can be easily verified that ${\bf{D}} \times ({\bf{F}}^{(\{1,2\})}_{[r_{\{1,2\}}]} | {\bf{\hat{F}}}^{(j)}) \approx {\bf{F}}'_x$. Hence  using Lemma \ref{dgmatrix}, the code ${\bf{c}}^{(\{1,2\})}+{\bf{c}}^{(j)}$ sent by $\mathcal{S}_j$ is a valid code for the sub-problem with fitting matrix  ${\bf{F}}'_x$. This satisfies all receivers in $
	\mathcal{I}_{\{1,2\}}$. Consider the code ${\bf{c}}^{(\{1,2\})}+{\bf{c}}^{(j^c)}$ sent by  $\mathcal{S}_{j^c}$, where $j^c = \{1,2\} \setminus j$. All the receivers in $\mathcal{I}_{1}$ and $\mathcal{I}_{2}$ are satisfied by using both the transmissions to get ${\bf{c}}^{(j)}-{\bf{c}}^{(j^c)}=({\bf{c}}^{(\{1,2\})}+{\bf{c}}^{(j)})-({\bf{c}}^{(\{1,2\})}+{\bf{c}}^{(j^c)})$. This is a valid code for the sub-problem with fitting matrix ${\bf{\tilde{F}}}_x$. This establishes $(i)$.
	
	We now prove $(ii)$. Consider the TGICP $\mathcal{I}''$ obtained from $\mathcal{I}$ by having fully-paticipated interactions 
	between $\mathcal{I}_{j}$ and $\mathcal{I}_{\{1,2\}}$. Then, we have $l^*_q(\mathcal{I}) \geq l^*_q(\mathcal{I}'')=\text{mrk}_{q}({\bf{F}}_{x}^{(j^c)})+\text{mrk}_{q}({\bf{F}}_{x}^{(\{1,2\})})$ using the result for Case II-C or II-D depending on $j$. The matching upper bound is given by $(i)$.
\end{proof}

We now illustrate Theorem \ref{thcase2e2} with an example.

\begin{exmp}
	Consider the partitioned fitting matrix ${\bf{F}}_x$ related to the TGICP $\mathcal{I}$ as given below. We choose $\mathbb{F}_2$. Let ${\bf{x}}_{\mathcal{M}_1}=\{{\bf{x}}_1,{\bf{x}}_2,{\bf{x}}_3,{\bf{x}}_4,{\bf{x}}_5,{\bf{x}}_8,{\bf{x}}_9\}$, ${\bf{x}}_{\mathcal{M}_2}=\{{\bf{x}}_6,{\bf{x}}_7,{\bf{x}}_8,{\bf{x}}_9\}$.
	\[ {\bf{F}}_{x} = \left(
	\begin{array}{ccccc|cc|cc}
	1 & x & x & 0 & 0 & x & x & 0 & 0\\
	0 & 1 & 0 & x & 0 & 0 & 0 & 0 & 0\\
	0 & 0 & 1 & 0 & x & 0 & 0 & 0 & 0\\
	x & 0 & 0 & 1 & x & x & x & 0 & 0\\
	x & x & 0 & 0 & 1 & x & x & 0 & 0\\
	\hline
	x & x & x & 0 & 0 & 1 & x & 0 & 0\\
	x & x & x & x & x & x & 1 & 0 & 0\\
	\hline
	x & 0 & x & x & 0 & x & x & 1 & x\\
	0 & x & 0 & 0 & x & 0 & 0 & 0 & 1\\
	\end{array} \right).
	\]
	Note that $\text{mrk}_q({\bf{F}}^{(1)}_x)=3$, $\text{mrk}_q({\bf{F}}^{(2)}_x)=1$, $\text{mrk}_q({\bf{F}}^{(\{1,2\})}_x)=2$, and the interaction digraph is $\mathcal{H}_{59}$. As $\text{mrk}_q({\bf{F}}^{(\{1,2\})}_x) \geq \text{mrk}_q({\bf{F}}^{(2)}_x)$, and $\text{mrk}_q({\bf{F}}^{(\{1,2\})}_x) < \text{mrk}_q({\bf{F}}^{(1)}_x)$, we have $j=2$.
	With the completions of the fitting matrices of the sub-problems of the TGICP $\mathcal{I}$ given below, we see that the SGICP with the  fitting matrix ${\bf{\tilde{F}}}_x$ (given in Theorem \ref{thcase2e2}) is a joint extension given in Theorem \ref{th1}. 
	\[	
	{\bf{F}}^{(1)} = \left(
	\begin{array}{ccccc}
	1 & 1 & 1 & 0 & 0\\
	0 & 1 & 0 & 1 & 0\\
	0 & 0 & 1 & 0 & 1\\	
	\hline
	1 & 0 & 0 & 1 & 1\\
	1 & 1 & 0 & 0 & 1\\
	\end{array} \right),
	{\bf{F}}^{(\{1,2\})} = \left(
	\begin{array}{cc}
	1 & 0\\
	0 & 1\\
	\end{array} \right).
	\]		
	\[	
	{\bf{F}}^{(2)} = \left(
	\begin{array}{cc}
	1 & 1 \\
	\hline
	1 & 1 \\
	\end{array} \right).
	\]
	Note also that ${\bf{\hat{F}}}^{(2)}=\left( \begin{array}{cc}
	1 & 1\\
	0 & 0\\
	\end{array} \right)$, and ${\bf{P}}^{\{1,2\}}$ is an empty matrix and all the conditions in Theorem \ref{thcase2e2} are satisfied. 
 Hence, $l^*_q(\mathcal{I})=5$. An optimal scalar linear code given according to Theorem \ref{thcase2e2} is : $\mathcal{S}_1$ sends $(~{\bf{x}}_1+{\bf{x}}_2+{\bf{x}}_3, ~{\bf{x}}_2+{\bf{x}}_4,~{\bf{x}}_3+{\bf{x}}_5~)$, and $\mathcal{S}_2$ sends $({\bf{x}}_6+{\bf{x}}_7+{\bf{x}}_8,~{\bf{x}}_9)$.	
\end{exmp}

The following corollaries establish the optimal scalar linear code lengths for some classes of the TGICP with partially-participated interactions belonging to Case II-E. The proofs follow on similar lines as that of Theorem \ref{thcase2e2}.

\begin{cor}
	Consider the TUICP $\mathcal{I}$ belonging to Case II-E with the interaction digraph being $\mathcal{H}_{61}$. Let $r_{\mathcal{A}}=\text{mrk}_q({\bf{F}}_x^{(\mathcal{A})})$, for non-empty $\mathcal{A} \subseteq \{1,2\}$, and $r_2 \geq r_{\{1,2\}} \geq r_1$. Also, the SGICP $\mathcal{\tilde{I}}$ given by the partitioned fitting matrix ${\bf{\tilde{F}}}_x$ in (\ref{eqcase2d2}) be a joint extension of $\mathcal{I}_1$ and $\mathcal{I}_{2}$ given by Theorem \ref{th1}, with optimal completions ${\bf{F}}^{(1)}$ and ${\bf{F}}^{(2)}$.
	Let ${\bf{F}}^{(\{1,2\})} \approx {\bf{F}}^{(\{1,2\})}_x$ be such that the first $r_{\{1,2\}}$ rows of ${\bf{F}}^{(\{1,2\})}$ span $\langle {\bf{F}}^{(\{1,2\})} \rangle$.  Let ${\bf{A}}^{(\{1,2\},1)}_{x_0}$ be any matrix such that ${\bf{A}}^{(\{1,2\},1)} \approx {\bf{A}}^{(\{1,2\},1)}_{x_0}$ where  
	${\bf{A}}^{(\{1,2\},1)}=\left( \begin{array}{c}  {\bf{\hat{F}}}^{(1)} \\
	\hline
	{\bf{P}}^{(\{1,2\})}{\bf{\hat{F}}}^{(1)}
	\end{array} \right)$, where 
	${\bf{\hat{F}}}^{(1)} =
	\left( \begin{array}{c} {\bf{F}}^{(1)}_{[r_1]}\\
	\hline
	{\bf{0}}_{(r_{\{1,2\}}-r_1) \times m_1}\\
	\end{array}
	\right)$, and ${\bf{P}}^{(\{1,2\})}$ be the $(n_{\{1,2\}}-r_{\{1,2\}}) \times r_{\{1,2\}}$ matrix such that the last $(n_{\{1,2\}}-r_{\{1,2\}})$ rows of ${\bf{F}}^{(\{1,2\})}$ is given by  ${\bf{P}}^{(\{1,2\})}{\bf{F}}^{(\{1,2\})}_{[r_{\{1,2\}}]}$. Let ${\bf{A}}^{(2,\{1,2\})}_{x_0}$ be any matrix such that ${\bf{A}}^{(2,\{1,2\})} \approx {\bf{A}}^{(2,\{1,2\})}_{x_0}$ where  
	${\bf{A}}^{(2,\{1,2\})}=\left( \begin{array}{c}  {\bf{\hat{F}}}^{(\{1,2\})} \\
	\hline
	{\bf{P}}^{(2)}{\bf{\hat{F}}}^{(\{1,2\})}
	\end{array} \right)$, where 
	${\bf{\hat{F}}}^{(\{1,2\})} =
	\left( \begin{array}{c} {\bf{F}}^{(\{1,2\})}_{[r_{\{1,2\}}]}\\
	\hline
	{\bf{0}}_{(r_2-r_{\{1,2\}}) \times m_{\{1,2\}}}\\
	\end{array}
	\right)$, and ${\bf{P}}^{(2)}$ be the $(n_{2}-r_{2}) \times r_{2}$ matrix such that the last $(n_{2}-r_{2})$ rows of ${\bf{F}}^{(2)}$ is given by  ${\bf{P}}^{(2)}{\bf{F}}^{(2)}_{[r_{2}]}$. Then we have, $l_{q}^*(\mathcal{I}) = \text{mrk}_{q}({\bf{F}}^{(2)}_x)+\text{mrk}_{q}({\bf{F}}^{(\{1,2\})}_x)$.\\
	\label{corcase2e61}
\end{cor}

\begin{cor}
	Consider the TUICP $\mathcal{I}$ belonging to Case II-E with the interaction digraph being $\mathcal{H}_{62}$. Let $r_{\mathcal{A}}=\text{mrk}_q({\bf{F}}_x^{(\mathcal{A})})$, for non-empty $\mathcal{A} \subseteq \{1,2\}$, and $r_1 \geq r_{\{1,2\}} \geq r_2$. Also, the SGICP $\mathcal{\tilde{I}}$ given by the partitioned fitting matrix ${\bf{\tilde{F}}}_x$ in (\ref{eqcase2d2}) be a joint extension of $\mathcal{I}_1$ and $\mathcal{I}_{2}$ given by Theorem \ref{th1}, with optimal completions ${\bf{F}}^{(1)}$ and ${\bf{F}}^{(2)}$.
	Let ${\bf{F}}^{(\{1,2\})} \approx {\bf{F}}^{(\{1,2\})}_x$ be such that the first $r_{\{1,2\}}$ rows of ${\bf{F}}^{(\{1,2\})}$ span $\langle {\bf{F}}^{(\{1,2\})} \rangle$.  Let ${\bf{A}}^{(1,\{1,2\})}_{x_0}$ be any matrix such that ${\bf{A}}^{(1,\{1,2\})} \approx {\bf{A}}^{(1,\{1,2\})}_{x_0}$ where  
	${\bf{A}}^{(1,\{1,2\})}=\left( \begin{array}{c}  {\bf{\hat{F}}}^{(\{1,2\})} \\
	\hline
	{\bf{P}}^{(1)}{\bf{\hat{F}}}^{(\{1,2\})}
	\end{array} \right)$, where 
	${\bf{\hat{F}}}^{(\{1,2\})} =
	\left( \begin{array}{c} {\bf{F}}^{(\{1,2\})}_{[r_{\{1,2\}}]}\\
	\hline
	{\bf{0}}_{(r_1-r_{\{1,2\}}) \times m_{\{1,2\}}}\\
	\end{array}
	\right)$, and ${\bf{P}}^{(1)}$ be the $(n_{1}-r_{1}) \times r_{1}$ matrix such that the last $(n_{1}-r_{1})$ rows of ${\bf{F}}^{(1)}$ is given by  ${\bf{P}}^{(1)}{\bf{F}}^{(1)}_{[r_{1}]}$. Let ${\bf{A}}^{(\{1,2\},2)}_{x_0}$ be any matrix such that ${\bf{A}}^{(\{1,2\},2)} \approx {\bf{A}}^{(\{1,2\},2)}_{x_0}$ where  
	${\bf{A}}^{(\{1,2\},2)}=\left( \begin{array}{c}  {\bf{\hat{F}}}^{(2)} \\
	\hline
	{\bf{P}}^{(\{1,2\})}{\bf{\hat{F}}}^{(2)}
	\end{array} \right)$, where 
	${\bf{\hat{F}}}^{(2)} =
	\left( \begin{array}{c} {\bf{F}}^{(2)}_{[r_2]}\\
	\hline
	{\bf{0}}_{(r_{\{1,2\}}-r_2) \times m_2}\\
	\end{array}
	\right)$, and ${\bf{P}}^{(\{1,2\})}$ be the $(n_{\{1,2\}}-r_{\{1,2\}}) \times r_{\{1,2\}}$ matrix such that the last $(n_{\{1,2\}}-r_{\{1,2\}})$ rows of ${\bf{F}}^{(\{1,2\})}$ is given by  ${\bf{P}}^{(\{1,2\})}{\bf{F}}^{(\{1,2\})}_{[r_{\{1,2\}}]}$. Then we have, $l_{q}^*(\mathcal{I}) = \text{mrk}_{q}({\bf{F}}^{(1)}_x)+\text{mrk}_{q}({\bf{F}}^{(\{1,2\})}_x)$.\\
	\label{corcase2e62}
\end{cor}

We now illustrate Corollary \ref{corcase2e62} with an example.

\begin{exmp}
	Consider the partitioned fitting matrix ${\bf{F}}_x$ related to the TGICP $\mathcal{I}$ as given below. We choose $\mathbb{F}_2$. Let ${\bf{x}}_{\mathcal{M}_1}=\{{\bf{x}}_1,{\bf{x}}_2,{\bf{x}}_3,{\bf{x}}_4,{\bf{x}}_5,{\bf{x}}_8,{\bf{x}}_9\}$, ${\bf{x}}_{\mathcal{M}_2}=\{{\bf{x}}_6,{\bf{x}}_7,{\bf{x}}_8,{\bf{x}}_9\}$.
	\[ {\bf{F}}_{x} = \left(
	\begin{array}{ccccc|cc|cc}
	1 & x & x & 0 & 0 & x & x & x & 0\\
	0 & 1 & 0 & x & 0 & 0 & 0 & 0 & x\\
	0 & 0 & 1 & x & x & 0 & 0 & 0 & 0\\
	x & 0 & 0 & 1 & x & x & x & x & x\\
	x & x & 0 & 0 & 1 & x & x & x & x\\
	\hline
	x & x & x & 0 & 0 & 1 & x & 0 & 0\\
	x & x & x & x & x & x & 1 & 0 & 0\\
	\hline
	0 & 0 & 0 & 0 & 0 & x & x & 1 & x\\
	0 & 0 & 0 & 0 & 0 & 0 & 0 & 0 & 1\\
	\end{array} \right).
	\]
	Note that $\text{mrk}_q({\bf{F}}^{(1)}_x)=3$, $\text{mrk}_q({\bf{F}}^{(2)}_x)=1$, $\text{mrk}_q({\bf{F}}^{(\{1,2\})}_x)=2$, with $r_1 \geq r_{\{1,2\}} \geq r_2$, and the interaction digraph is $\mathcal{H}_{62}$. With the completions of the fitting matrices of the sub-problems of the TGICP $\mathcal{I}$ given below, we see that the SGICP with the  fitting matrix ${\bf{\tilde{F}}}_x$ (given in Corollary \ref{corcase2e62}) is a joint extension given in Theorem \ref{th1}. 
	\[	
	{\bf{F}}^{(1)} = \left(
	\begin{array}{ccccc}
	1 & 1 & 1 & 0 & 0\\
	0 & 1 & 0 & 1 & 0\\
	0 & 0 & 1 & 0 & 1\\	
	\hline
	1 & 0 & 0 & 1 & 1\\
	1 & 1 & 0 & 0 & 1\\
	\end{array} \right),
	{\bf{F}}^{(\{1,2\})} = \left(
	\begin{array}{cc}
	1 & 0\\
	0 & 1\\
	\end{array} \right).
	\]		
	\[	
	{\bf{F}}^{(2)} = \left(
	\begin{array}{cc}
	1 & 1 \\
	\hline
	1 & 1 \\
	\end{array} \right).
	\]
	Note also that ${\bf{\hat{F}}}^{(2)}=\left( \begin{array}{cc}
	1 & 1\\
	0 & 0\\
	\end{array} \right)$, and ${\bf{\hat{F}}}^{(\{1,2\})}=\left( \begin{array}{cc}
	1 & 0\\
	0 & 1\\
	0 & 0\\
	\end{array} \right)$, and all the conditions in Corollary \ref{corcase2e62} are satisfied. 
	Hence, $l^*_q(\mathcal{I})=5$. An optimal scalar linear code given according to Corollary \ref{corcase2e62} is : $\mathcal{S}_1$ sends $(~{\bf{x}}_1+{\bf{x}}_2+{\bf{x}}_3+{\bf{x}}_8, ~{\bf{x}}_2+{\bf{x}}_4+{\bf{x}}_9,~{\bf{x}}_3+{\bf{x}}_5~)$, and $\mathcal{S}_2$ sends $({\bf{x}}_6+{\bf{x}}_7+{\bf{x}}_8,~{\bf{x}}_9)$.	
\end{exmp}
%%%%%%%%%%%%%%%%%%%%%%%%%%%%%%%%%%
\section{Conclusion}   
\par In this work, we construct optimal scalar linear codes for some classes of the TGICP using those of the sub-problems. The notion of joint extensions of a finite number of SGICPs is introduced and exploited in the code constructions. A class of joint extensions has been identified where optimal scalar linear codes for jointly extended problems can be constructed using those of the sub-problems.\\   
%%%%%%%%%%%%%%%%
%\vspace*{5.0cm}
%%%%%%%%%%%%%%%%%%%%%%%%%%%%%%%%
%%%%%%%%%%%%%%%%%%%%%%%%%%%%%%%%
\section*{Acknowledgment}
This work was supported partly by the Science and Engineering Research Board (SERB) of Department of Science and Technology (DST), Government of India, through J.C. Bose National Fellowship to B. Sundar Rajan.

\end {document}